\documentclass[11pt]{article}

\usepackage{amssymb}           
\usepackage{amsthm}            
\usepackage{amsmath,mathtools} 
\usepackage{tikz}
  \usetikzlibrary{chains,positioning,scopes,arrows,plotmarks} 
\usepackage{comment}           
\usepackage{array}             
\usepackage{tabto}             
\usepackage{xspace}            
\usepackage{fullpage}

\newcommand{\qqed}{}
\newcommand{\BZ}{\text{{\sc Bz}}\xspace}
\newcommand{\EXPN}{{\sc Expansion}\xspace}
\newcommand{\Rae}{{\sc Rae}\xspace}

\newcommand{\RaeM}{\text{{\sc Rae}}\xspace}
\newcommand{\SAT}{{\sc Satisfaction}\xspace}
\newcommand{\SatM}{\text{{\sc Sat}}\xspace}
\newcommand{\VC}{$\#\frac{2}{3}${\sc -VC}\xspace}

\newcommand{\cL}{{\cal L}}
\newcommand{\cM}{{\cal M}}
\newcommand{\cP}{{\cal P}}
\newcommand{\cW}{{\cal W}}
\newcommand{\tF}{{\tt F}}
\newcommand{\tI}{{\tt I}}
\newcommand{\tL}{{\tt L}}

\newcommand{\tR}{{\tt R}}
\newcommand{\tFI}{{\tt FI}}
\newcommand{\tRF}{{\tt RF}}
\newcommand{\tLR}{{\tt LR}}
\newcommand{\OLF}{{OLF}\xspace}
\newcommand{\oOLF}{{odd-OLF}\xspace}

\newcommand{\gOLF}{{gOLF}\xspace}

\newcommand{\remove}[1]{}

\newtheorem{definition}{Definition}
\newtheorem{example}{Example}
\newtheorem{lemma}{Lemma}
\newtheorem{theorem}{Theorem}
\newtheorem{corollary}{Corollary}

\begin{document}

%

\title{Measuring satisfaction in societies\\ with opinion leaders and mediators}

\author{Xavier Molinero\\ Department of Mathematics\\
Technical University of Catalonia, Manresa, Spain. \\  
\texttt{xavier.molinero@upc.edu}\\
\ \\
Fabi\'an Riquelme \\CITIAPS\\ University of Santiago, Santiago, Chile.\\ 
\texttt{fabian.riquelme.c@usach.cl} \\
\ \\ Maria Serna\\Computer Science Department.\\ 
Technical University of Catalonia, Barcelona, Spain.\\
\texttt{mjserna@cs.upc.edu}}
\date{}

\maketitle

\thispagestyle{empty}
\begin{abstract}
An opinion leader-follower model (\OLF) is a two-action collective deci\-sion-making model for societies, in which three kinds of actors are considered: {\em opinion leaders}, {\em followers}, and {\em independent actors}. In \OLF the initial decision of the followers can be modified by the influence of the leaders. Once the final decision is set, a collective decision is taken applying the simple majority rule~\cite{BRS11}. 
We consider a generalization of \OLF, the \gOLF models which allow collective decision taken by rules different from the single majority rule.   Inspired in this model we define two new families of collective decision-making models associated with influence games~\cite{MRS15}. We define the \emph{oblivious} and \emph{non-oblivious influence} models. We show that \gOLF models are non-oblivious influence models played on a two layered bipartite influence graph. 

Together with  \OLF models,  the {\em satisfaction} measure was introduced and studied.  We analyze the computational complexity of the satisfaction measure for \gOLF models and the other collective decision-making models introduced in the paper. 
We show that computing the satisfaction measure is \#P-hard in all the considered models. 
On the other hand, we provide two subfamilies of decision models in which the satisfaction measure can be computed in polynomial time.   

Exploiting the relationship with  influence games, we can relate  the satisfaction measure with  the Rae index of an associated simple game. The Rae index is closely related to the Banzhaf. Thus, our results also extend the families of simple games for which computing the Rae index and the Banzhaf value is  computationally hard.

\smallskip
\noindent
\textbf{Keywords}
Collective decision-making model. Opinion leader-follower model. Simple game. Influence game. Satisfaction measure. Banzhaf value. Computational complexity.
\end{abstract}


\section{Introduction}

Opinion leadership is a well known and established model for communication in sociology and marketing. It comes from the {\em two-step flow of communication} theory proposed in the 1940s~\cite{LBG68}. This theory recognizes the existence of collective decision-making situations in societies formed by actors called {\em opinion leaders}, who exert influence over other kind of actors called the {\em followers}, resulting in a two-step decision process~\cite{LBG68,KL55}. In the first step of the process, all actors receive information from the environment, generating their own decisions; in the second step, a flow of influence from some actors over others is able to change the choices of some of them~\cite{Tro66}.
Following those ideas an {\em opinion leader-follower model} (\OLF)  was introduced in~\cite{BRS11}. The subjacent influence structure is based on a  society with opinion leaders, followers and {\em independent actors}. This latter kind of actors neither can influence nor can be influenced by other members of the society. The collective decision-making model includes a procedure to reach an individual final decision, from a given initial decision of the participants. Finally, a  global decision is taken by applying  the single majority rule to the final decisions of the actors which restricts the society to have an odd number of participants. 

In this paper we are interested in extending \OLF models in order to incorporate more complex scenarios in the decision-making model. In general, simple games are used to formulate the situations in which the actors or players have to decide about one alternative. Simple games were firstly introduced in 1944 by von Neumann and Morgenstern~\cite{vNM44} 
and constitute the classic way for such models. Briefly speaking, a simple game is determined by its {\em winning coalitions}, i.e., those sets of actors that can force a ``yes'' decision. 
Nevertheless, the way in which the influence is exerted among leaders and followers in an \OLF suggests that another formalism based on graphs might provide the necessary insights, so  we explore the possibility of defining  decision-making models based on  {\em influence games} ~\cite{MRS15}. Influence games are  based on social networks that incorporate a procedure of influence spreading. The winning coalitions are determined depending on the level of influence that the coalitions can exert in a social network.  Influence games are defined through a more complex graph structure among the actors than the structures of the \OLF. Here the society is modeled by an \emph{influence graph} and it is assumed that influence spreads according to the linear threshold model~\cite{Gra78,Sch78,KKT03}. Influence games are general enough to capture the complete class of simple games, since every simple game can be represented by an influence game~\cite{MRS15}. However, the size of the graph representing a simple game might be exponential in the number of players. 

We first define an extension of the family of \OLF, the \emph{generalized} \OLF (\gOLF),  by allowing, in the final decision mechanisms, a rule different from the simple majority rule.  We  introduce two new kinds of collective decision-making models based on the influence spread mechanism of influence games. In the {\em oblivious influence models}, the initial decision of the followers is not taken into account and a negative initial decision of the actors is assumed. In the {\em non-oblivious influence models}, the initial decision of the followers is taken into consideration in a similar way as for \OLF. We also consider a subfamily of \gOLF, the \oOLF. This subfamily plays an important role as it lies in  the intersection of oblivious and non-oblivious models. The inclusion relationship among the introduced models is explained later, but for now it is depicted in Figure~\ref{fig:submodels}. 

\begin{figure}[t]
\centering
\scalebox{0.8}{
\begin{tikzpicture}
\begin{scope}[fill opacity=0.5]
\draw[rounded corners=3ex] (-1,-.5) rectangle (250pt,190pt);
\draw[fill=gray!10, draw=black] (5.5,2.6) circle (2.6);
\draw[fill=gray!20, draw=black] (2.5,2.6) circle (2.6);
\draw[fill=gray!30, draw=black] (2.7,2.6) circle (1.8);
\draw[fill=gray!50, draw=black] (3.7,2.6) circle (.7);
\draw[fill=gray!50, draw=black] (3,2) circle (.5);
\node at (2,5.8)   (S) {\textbf{collective decision-making models}};
\node at (6.6,2.8) (D) {\textbf{oblivious influence}};
\node at (2.5,4.6) (R) {\textbf{\small{non-oblivious influence}}};
\node at (2.1,2.8) (W) {\textbf{\footnotesize{\gOLF}}};
\node at (3.7,2.6) (W) {\textbf{\footnotesize{\oOLF}}};
\node at (3,2) (W) {\textbf{\footnotesize{\OLF}}};
\end{scope}
\end{tikzpicture}
}\caption{Inclusion relationship between subfamilies of collective decision-making models.\label{fig:submodels}}
\end{figure}
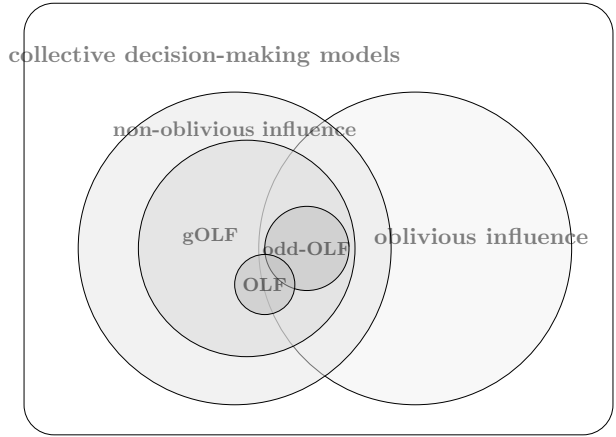

Motivated by the theoretical study of the effects that collective decision-making can have on the participants, a {\em satisfaction} measure 
was defined for an \OLF in van den Brink et al.~\cite{BRS11}. The satisfaction of an actor refers to the number of initial decisions for which the collective decision coincides with the initial decision of the actor. 
In this paper we are interested in, besides the definition of the decision models,  the computational complexity of the satisfaction measure.   
We show that the satisfaction problem is \#P-hard for \oOLF. It is interesting to note that the influence graph corresponding to those models is  formed by a simple two layered bipartite graph. Our hardness proof comes from a reduction from the \VC problem, i.e., the problem of computing the number of vertex covers with exactly $\frac{2}{3}n$ vertices. This problem is equivalent to the problem of counting independent sets of size $\frac{1}{3}n$ which was shown to be \#P-hard, by a reduction from the well known \#P-complete $\#3${\sc-SAT} problem, the counting version of the 3-satisfiability problem~\cite{Hof10}. 

We are also interested in analyzing what are the influence graphs for which the satisfaction measure can be computed in polynomial time. 
As the satisfaction problem is hard for two layered bipartite graphs, there is not much room for finding tractable subfamilies. We explore the possibility of tractable cases in multilayered bipartite graphs where the connection among consecutive layers is the disjoint union of complete bipartite graphs. Those structures are able to represent a ``more-than-two-step flow of communication''. In multilayered bipartite graphs we have an additional set of actors called the {\em mediators}. Those actors  can be influenced by the opinion leaders and may influence the followers. Thus we extend \OLF models by allowing an intermediate set of actors to play the role of mediators between leaders and followers. The first family, the {\em strong hierarchical influence graphs}, is formed by layered bipartite graphs in which influence is exerted in an all-to-all fashion following a hierarchical structure. In this case the mediators are interposed between the leaders and the followers. In the second family, the {\em star influence graphs}, we have only one mediator, but in this case we allow a two way interaction between the mediator and some opinion leaders, modeling a natural mediation schema occurring in society.
We show  for those two subfamilies of bipartite influence graphs that  the satisfaction measure can be computed in polynomial time, for both the oblivious and non-oblivious associated decision-making models. 

Interestingly enough, every collective decision-making model considered in this paper is monotonic and thus can be reinterpreted as a simple game. Under this interpretation, we are able to show  that the satisfaction measure coincides with the {\em Rae index}. The Rae index was introduced by Douglas W. Rae~\cite{Rae69} for anonymous games and afterwards it was applied by Dubey and Shapley~\cite{DS79} to simple games. In the context of simple games, Dubey and Shapley~\cite{DS79} established an affine-linear relation between the Rae index and the Banzhaf value~\cite{LMF06}, establishing a computational complexity equivalence among their computation. Thus our results on the complexity of the satisfaction problem, both positive and negative, apply to the computation of the Rae index and the Banzhaf value. Computing the Banzhaf value is polynomial time solvable for simple games represented by the set of winning coalitions, but is \#P-complete for simple games represented by the set of minimal\footnote{A minimal winning coalition is a winning coalition such that by removing any player we obtain a losing coalition. A maximal losing coalition is a losing coalition such that by adding any player we obtain a winning coalition.}  winning coalitions~\cite{Azi08}. The problem is also known to be \#P-hard for weighted voting games~\cite{DP94} and influence games~\cite{MRS15}. For the case of influence games, the known \#P-hardness result leaves open the question on whether the problem is easy or hard when the influence graph is restricted to be bipartite. Thus, our results extend the subfamilies of simple games for which the complexity of the computation of the Banzhaf value is known.

The paper is organized as follows. Section~\ref{sec:prelim} presents the definitions of \OLF,  the satisfaction measure, simple and influence games. In Section~\ref{sec:oblivious} we introduce the {\em oblivious} and {\em non-oblivious influence models} and show  some of their basic properties. In Section~\ref{sec:SAT} we prove the hardness of the satisfaction problem for \oOLF, which implies hardness for the other three considered models. Sections~\ref{sec:SATpol} and ~\ref{sec:Star-Inf} are devoted to the definition of strong hierarchical influence graphs and star influence graphs, respectively, and to the design of algorithms to solve the satisfaction problem in polynomial time for oblivious and non-oblivious influence models on those graph classes. We finalize this paper stating some conclusions and open problems.

\section{Preliminary concepts}
\label{sec:prelim}

In this section we introduce the necessary  definitions and concepts, such as \OLF, the satisfaction measure, simple and influence games. Before stating the definitions, we fix some notation for graphs and sets.

All the graphs considered in this paper are directed, unless otherwise stated, without loops and multiple edges. We use standard notation for graphs: $G=(V,E)$ is a directed graph, $V(G)$ denotes the vertex set, $E(G)$ is the edge set, and $n$ denotes the number of vertices $|V|$. We use simply $V$ and $E$ when there is no risk of confusion.
For $i\in V$, $S_G(i)=\{j\in V\mid (i,j)\in E\}$ denotes the set of {\em successors} of $i$, and $P_G(i)=\{j\in V\mid (j,i)\in E\}$ the set of {\em predecessors} of $i$. We extend this notation to vertex subsets, so for $X\subseteq V$, $S_G(X)=\{i\in V\mid \exists j\in X, i\in S_G(j)\}$ and $P_G(X)=\{i\in V\mid \exists j\in X, i\in P_G(j)\}$ denote the set of successors and predecessors of all elements that belong to $X$, respectively. Finally, let  $\delta^-(i)=|P_G(i)|$ and $\delta^+(i)=|S_G(i)|$ denote the in-degree and the out-degree of the vertex $i$, respectively.
A {\em two layered bipartite graph} is a bipartite graph $G=(V_1,V_2,E)$ with $V(G)=V_1\cup V_2$ and $E\subseteq V_1\times V_2$, i.e., so that for $i\in V_1$, $P_G(i)=\emptyset$ and,
for $i\in V_2$, $S_G(i)=\emptyset$.  By $I_a$, as usual, we denote a graph that is formed by $a$ isolated vertices.
Given $G=(V,E)$ and $X\subseteq V$, $G[X]$ denotes the subgraph induced by $X$ and $G\setminus X$ denotes the subgraph induced by $V\setminus{}X$, i.e., $G\setminus{}X = G[V\setminus X]$.

As usual, given a finite set $N$, $\cP(N)$ denotes its power set.
A family of subsets $\cW\subseteq\cP(N)$ is said to be {\em monotonic} when, for $X\in\cW$, if $X\subseteq Z$, then $Z\in\cW$.
In some cases we need to establish a relation between vectors $x=(x_1,\ldots,x_n)\in\{0,1\}^n$ and sets $X\subseteq\{1,\dots,n\}$. For doing so, we use the notation $X(x)=\{0\le{}i\le{}n\mid x_i=1\}$, and $x(X)=(x_1,\dots, x_n)$ with $x_i=1$ if and only if $i\in X$, and $x_i=0$ otherwise.

Every collective decision-making model $\cM$ considered in this paper is  (or can be) defined on  a  weighted digraph on which a  {\em collective decision function} is defined. 
The  actors (vertices) initially  can choose among  two alternatives, $1$ or $0$.  Given a set of $n$ actors, the initial decision is represented by an {\em initial decision vector} $x\in\{0,1\}^n$. The initial decisions of the actors, due to the interactions among the participants, and according to the model, can change giving rise to a {\em final decision vector} of the actors. We usually denote this final decision vector as $c^{\cM}(x)\in\{0,1\}^n$.  Finally,  the collective decision function $C_\cM:\{0,1\}^n\to\{0,1\}$, that depends on $c^\cM(x)$, assigns a final decision to the system. In the following subsections we formally define those vectors and functions for each of the considered decision-making models.  In general, we  drop the explicit reference to $\cM$ when the model is clear from the context.

Note that the collective decision function is usually defined through a decision process on a graph.
It may include many parameters, therefore its computational complexity might be  {\em high}. Nevertheless,
the models considered in this paper have their collective decision functions computable in polynomial time.

\subsection{The opinion leader-follower  model and the satisfaction measure}

The \OLF  was introduced by van den Brink et al.~\cite{BRS11}. As we mention before, the model  considers three kind of actors: {\em opinion leaders}, {\em followers} and {\em independent actors}.  We provide a formal definition of our generalization of \OLF. 

\begin{definition}\label{def1_2}
A {\em generalized opinion leader-follower model (\gOLF)} is a triple $\cM=(G,r,q)$ where $G=(V,E)$ is a two layered bipartite digraph that represents the actors' relations.   The {\em fraction value} $r$, $1/2\leq r\leq 1$, is a rational number. The {\em quota} $q$, $0<q\leq n$, is a natural number.  The collective decision function $C$ is defined as follows.
Let $x\in\{0,1\}^n$ be an initial decision vector, then the final decision vector $c=c^\cM(x)$ has the following components, for $1\leq i\leq n$:
$$c_i=\begin{cases}
   1 & \text{if  $|\{j\in P_G(i)\mid x_j=1\}|\geq \lceil r\cdot|P_G(i)|\rceil$}\\
& \qquad\text{ and  $|\{j\in P_G(i)\mid x_j=0\}|  <  \lceil r\cdot|P_G(i)|\rceil$}\\
   0 & \text{if $|\{j\in P_G(i)\mid x_j=0\}|\geq \lceil r\cdot|P_G(i)|\rceil$}\\
&\qquad \text{  and
                $|\{j\in P_G(i)\mid x_j=1\}|  <  \lceil r\cdot|P_G(i)|\rceil$}\\
 x_i & \text{otherwise.}\\
\end{cases}
$$
Finally,  the collective decision function $C_\cM:\{0,1\}^n \to \{0,1\}$ is defined as
$$C_\cM(x) =
\begin{cases}
 1 &  \mbox{if } |\{i\in V\mid c_i(x)=1\}| \geq q\\
 0 &  \mbox{otherwise}.\\
\end{cases}
$$
\end{definition}

Note that the values $c^\cM_i$, for $i\in V$,  and  $C_\cM$ are well defined. 
In the definition of \OLF, $n$ was required to be odd and $q$ was fixed  to $(n+1)/2$~\cite{BRS11}.

Observe that  the set $V$ can be partitioned into three subsets:  
 the opinion leaders $\tL(G)=\{i\in V\mid P_G(i)=\emptyset\ \text{and}\ S_G(i)\neq\emptyset\}$, the followers $\tF(G)=\{i\in V\mid S_G(i)=\emptyset\ \text{and}\ P_G(i)\neq\emptyset\}$ and the independent actors $\tI(G)=V\setminus(\tL(G)\cup\tF(G))=\{i\in V\mid S_G(i)=\emptyset\ \text{and}\ P_G(i)=\emptyset\}$. We sometimes will  need the set $\tFI(G)=\tF(G)\cup \tI(G)$. 
Note that if $(i,j)\in E$ then $i\in\tL(G)$ and $j\in\tF(G)$.
When there is no risk of ambiguity, we simply use $S(i)$, $P(i)$, $\tI$, $\tL$ or $\tF$ omitting the corresponding graph $G$.
Finally, note also that $S(\tL)=\tF$, $P(\tF)=\tL$ and $S(\tI)=P(\tI)=\emptyset$.

Observe that both opinion leaders and independent actors always keep their own inclinations in the final decision vector, but a follower may take a final decision different from  their own initial inclination. In particular, for followers with an even number of predecessors,  when $r=1/2$ every tie arising in the predecessors' decision is broken by the follower's initial decision.
Also, the  restriction $1/2\leq r\leq 1$ could be replaced by  $0\leq r\leq 1$, without changing the model, as a follower having $r<1/2$ can be replaced by an independent actor. 
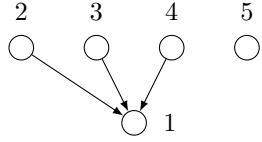
\begin{figure}[t]
\centering
\begin{tikzpicture}[every node/.style={circle,scale=0.9}, >=latex]
\node[draw](a) at (1.5,0.0)[label=right:$1$] {};
\node[draw](b) at (0.0,1.0)[label=above:$2$] {};
\node[draw](c) at (1.0,1.0)[label=above:$3$] {};
\node[draw](d) at (2.0,1.0)[label=above:$4$] {};
\node[draw](e) at (3.0,1.0)[label=above:$5$] {};
\draw[->] (b) to node {}(a);
\draw[->] (c) to node {}(a);
\draw[->] (d) to node {}(a);
\end{tikzpicture}
\caption{A two layered bipartite graph.\label{fig1}}
\end{figure}

\begin{example}\label{ex1}
Figure \ref{fig1} illustrates a two layered bipartite graph $G=(V,E)$.
Let $\cM=(G,1/2,3)$ be a \gOLF.
For the initial decision vectors $x=(0,1,1,0,0)$ and $y=(1,1,1,0,0)$, we get the same final decision vector, $c^\cM(x)=c^\cM(y)=(1,1,1,0,0)$,
and the same collective decision value, $C_\cM(x)=C_\cM(y)=1$.
\end{example}

Now we define the satisfaction measure. This measure was defined on the set of actors of the original \OLF to compare and contrast different structures of the model~\cite{BRS11}. However, the concept applies to any collective decision-making model.

\begin{definition}
Let $\cM$ be a collective decision-making model over a set of $n$ actors. For an initial decision vector $x\in\{0,1\}^n$, an actor $i$ is {\em satisfied} when $C_\cM(x)=x_i$. The {\em satisfaction} measure of the actor $i$ corresponds to the number of initial decision vectors for which the actor is satisfied, i.e.,
$$\SatM_\cM(i) = |\{x\in\{0,1\}^n\mid C(x)=x_i\}|.$$
\end{definition}

Associated to the satisfaction measure, we consider the following computational problem.
\begin{description}
\item{\SAT}
\item{Instance:} \tabto{10ex} A collective decision-making model $\cM$ and an actor $i$.
\item{Output:}   \tabto{10ex} $\SatM_\cM(i)$.
\end{description}

\subsection{Simple and influence games}
\label{ssec:sg+ig}

In this subsection we define simple and influence games. In the scenario of simple games, we follow notation from~\cite{TZ99} and for influence games from~\cite{MRS15}.

\begin{definition}\label{def:SG}
A {\em simple game} is a tuple $\Gamma=(N,\cW)$, where $N$ is a finite set of {\em players} and $\cW\subseteq\cP(N)$ is a monotonic family of subsets of $N$.
\end{definition}

The subsets of $N$ are called {\em coalitions}, the set $N$ is the {\em grand coalition} and each $X\in\cW$ is a {\em winning coalition}.
The complement of the family of winning coalitions is the family of {\em losing coalitions} $\cL$, i.e., $\cL=\cP(N)\setminus\cW$.
Any of those set families determine uniquely the game $\Gamma$ and constitute one of the usual forms of representation for simple games~\cite{TZ99}, although the size of the representation is not, in general, polynomial in the number of players (see also~\cite{MRS15b}).

In simple game theory, power indices are used to measure the relevance that a player has in the game. We recall here the definition of two classic power indices, the Banzhaf value~\cite{Pen46,Ban65} and the Rae index~\cite{Rae69}.

\begin{definition}
Let $\Gamma=(N,\cW)$ be a simple game. The {\em Banzhaf value} of player $i\in N$ is defined as
$$\BZ_\Gamma(i)=|\{X\in\cW\mid X\setminus\{i\}\notin\cW\}|.$$
The {\em Rae index} of player $i\in N$ is
$$\RaeM_\Gamma(i) = |\{X\in\cW\mid i\in X\}|+|\{X\notin\cW\mid i\notin X\}|.$$
\end{definition}

Dubey and Shapley~\cite{DS79} (see also~\cite{LMF06}) established an affine-linear relation between the Rae index and the Banzhaf value.
Let $\Gamma=(N,\cW)$ be a simple game and $i\in N$,

\begin{equation}\label{eq:Sat_Rae_Bz}
\RaeM_\Gamma(i) = 2^{n-1} + \BZ_\Gamma(i).
\end{equation}

Note that, for $i\in N$, as $\BZ_\Gamma(i)\geq 0$, $\RaeM_\Gamma(i)\geq 2^{n-1}$.

Influence games constitute a form of representation of simple games based on graphs. Influence games are based on the linear threshold model for influence spreading~\cite{Gra78,Sch78,KKT03}
and take into account the spread of influence phenomenon through social networks. In this setting a coalition wins if it can induce enough participants  to accept the proposal. Before defining influence games we introduce influence graphs and the activation process by which influence spreads in the network.

\begin{definition}
An {\em influence graph} is a tuple $(G,f)$, where $G=(V,E)$ is a directed graph and $f$ is a labeling function assigning to any vertex a non-negative rational value.
\end{definition}

Let $(G,f)$ be an influence graph and let $X\subseteq V$. The {\em activation process}, with initial activation $X$, at time $t$, $0\leq t \leq n$, activates a set of vertices $F^t(X)$ defined as follows
\begin{align*}
F^0(X) &=X\\
F^{t}(X)&= F^{t-1}(X)\cup \{i\in V \mid |P_G(i)\cap F^{t-1}(X)| \geq f(i)\}, \text{ for } 1\leq t\leq n.
\end{align*}
The {\em spread of influence} of $X$ in $(G,f)$ is the set $F(X)=F^n(X)$.

Observe that, in the activation process, a new vertex is activated whenever the number of vertices activated in the previous step pointing to it is greater or equal to its label. Thus, the labeling function quantifies the number of neighbors that each actor requires to be influenced. By definition, the activation process is monotonic; therefore, as there are at most $n$ vertices, there may be a value $k<n$ so that $F^k(X)=F^n(X)$.
Note also that given an initial activation set $X\subseteq V$, $F(X)$ can be computed in polynomial time.

\begin{example}\label{ex2}
Figure~\ref{fig2} shows the spread of influence in an influence graph starting from the initial activation set $X=\{2,3\}$.
After the first step we have $F^1(X)=\{1,2,3\}$, and after the second step, the last one, we have $F^2(X)=F(X)=\{1,2,3,5\}$.
\begin{figure}[t]
\centering
\begin{tikzpicture}[every node/.style={circle,scale=0.8}, >=latex]
\node[draw](a)         at (1.5,0.0)[label=right:$1$] {2};
\node[fill=blue!20](b) at (0.0,1.0)[label=above:$2$] {1};
\node[fill=blue!20](c) at (1.0,1.0)[label=above:$3$] {1};
\node[draw](d)         at (2.0,1.0)[label=above:$4$] {1};
\node[draw](e)         at (3.0,1.0)[label=above:$5$] {1};
\draw[->] (b) to node {}(a);
\draw[->] (c) to node {}(a);
\draw[->] (d) to node {}(a);
\draw[->] (a) to node {}(e);
\end{tikzpicture}
\qquad
\begin{tikzpicture}[every node/.style={circle,scale=0.8}, >=latex]
\node[fill=blue!20](a) at (1.5,0.0)[label=right:$1$] {2};
\node[fill=blue!20](b) at (0.0,1.0)[label=above:$2$] {1};
\node[fill=blue!20](c) at (1.0,1.0)[label=above:$3$] {1};
\node[draw](d)         at (2.0,1.0)[label=above:$4$] {1};
\node[draw](e)         at (3.0,1.0)[label=above:$5$] {1};
\draw[->] (b) to node {}(a);
\draw[->] (c) to node {}(a);
\draw[->] (d) to node {}(a);
\draw[->] (a) to node {}(e);
\end{tikzpicture}
\qquad
\begin{tikzpicture}[every node/.style={circle,scale=0.8}, >=latex]
\node[fill=blue!20](a) at (1.5,0.0)[label=right:$1$] {2};
\node[fill=blue!20](b) at (0.0,1.0)[label=above:$2$] {1};
\node[fill=blue!20](c) at (1.0,1.0)[label=above:$3$] {1};
\node[draw](d)         at (2.0,1.0)[label=above:$4$] {1};
\node[fill=blue!20](e) at (3.0,1.0)[label=above:$5$] {1};
\draw[->] (b) to node {}(a);
\draw[->] (c) to node {}(a);
\draw[->] (d) to node {}(a);
\draw[->] (a) to node {}(e);
\end{tikzpicture}
\caption{Spread of influence (the set of colored vertices)  starting from the initial activation $X=\{2,3\}$.\label{fig2}}
\end{figure}
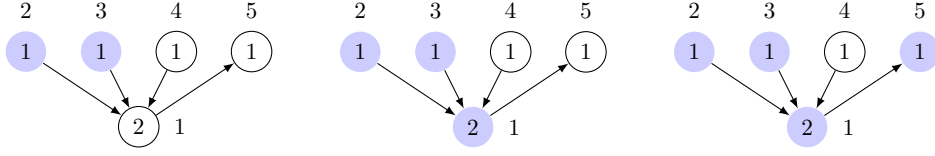
\end{example}

Influence games are simple games defined with the help of an influence graph where a coalition wins if it is able to influence a required number of the network participants.

\begin{definition}
\label{def:IG}
An {\em influence game} is a tuple $\Gamma=(G,f,q,N)$, where $(G,f)$ is an influence graph and $q$ is an integer, the {\em quota}, with  $0\leq q\leq n$.
The set of players is the set $N\subseteq V(G)$, and $X\subseteq V(G)$ is a winning coalition if and only if $|F(X\cap N)|\geq q$, otherwise $X$ is a losing coalition.
\end{definition}

Note that, in an influence game, the set of players $N$ can be a proper subset of the set of vertices in the graph $V$. It is easy to see that the set of winning coalitions of an influence game is monotonic, since for $X\subseteq V$ and $i\in V$, $|F(X\cap N)|\geq q$ entails $|F((X\cap N)\cup\{i\})|\geq q$. Therefore influence games are a subfamily of simple games. Moreover, any simple game admits a representation as influence game~\cite{MRS15}.

\section{Oblivious and non-oblivious influence models}
\label{sec:oblivious}

In this section we introduce two new collective decision-making models and  prove some basic properties relating the different models.
Following the definition of \gOLF, it seems natural to consider a collective decision-making model which uses  an influence graph to exert influence before the final decision is taken.  In our models the actors are the set of vertices $V$ of an influence graph, and the collective decision function takes into consideration the spread of influence process running on it. In such a setting a player represents either some sort of opinion leader or an independent actor, and a non-player a follower.  The main difference is that now players could be convinced of a change of opinion and take a final decision different from the initial one. This behaviour was not possible in the \gOLF model.  Based on this idea, we define two collective decision-making models.

When modeling collective decision models we require that $f(i)>0$, for $i\in V$.  In other words,  some positive level of influence to adopt an opinion is required. Nevertheless, for technical reasons, we keep open the possibility of having $f(i)=0$ for some actor $i$ in an influence graph. 

\begin{definition}
An {\em oblivious influence model} is a collective decision-making model $\cM=(G,f,q,N)$, where $(G,f,q,N)$ is an influence graph with positive labeling function and whose collective decision function is defined, for $x\in\{0,1\}^n$, as
$$C_\cM(x) = \begin{cases}
 1 & \mbox{if } |F(X(x)\cap N)|\geq q\\
 0 & \mbox{if } |F(X(x)\cap N)|<q.
\end{cases}
$$
\end{definition}

\begin{definition}
A {\em non-oblivious influence model} is a collective decision-making model $\cM=(G,f,q,N)$ with positive labeling function  whose collective decision function is defined as follows. For $x\in\{0,1\}^n$, $p_i(x)=|F(X(x)\cap N)\cap P(i)|$ and $q_i(x)=|P(i)\setminus F(X(x)\cap N)|$. For $i\in V(G)\setminus N$, we define the final decision vector $c=c^\cM(x)$ as
$$c_i = \begin{cases}
 1  & \mbox{if } p_i(x) \geq f(i) \text{ and } q_i(x) < f(i)\\
 0  & \mbox{if } q_i(x) \geq f(i) \text{ and } p_i(x) < f(i)\\
x_i & \mbox{otherwise}\\
\end{cases}
$$
and, for $i\in N$,
$$c_i = \begin{cases}
 1 & \mbox{if } i\in F(X(x))\\
 0 & \mbox{otherwise.}\\
\end{cases}
$$
Finally, the collective decision function, for $x\in\{0,1\}^n$, is defined as
$$C_\cM(x) = \begin{cases}
 1 & \mbox{if } |\{i\in V(G)\mid c_i=1\}|\geq q\\
 0 & \mbox{otherwise}.\\
\end{cases}
$$
\end{definition}

In order to avoid confusion, for an influence game $\Gamma=(G,f,q,N)$, we use $\cM^o(\Gamma)$ to denote the corresponding oblivious influence model and $\cM^n(\Gamma)$ to denote the corresponding non-oblivious influence model.

Observe that, in an influence game $\Gamma$  with $N=V$,   
for $\mathcal{M}^n(\Gamma)$, we have $|\{i\in V\mid c_i=1\}|=|\{i\in V\mid i\in F(X(x))\}|$, therefore  $C_{\cM^o}= C_{\cM^n}$.  
In oblivious influence models, as in influence games, the initial decision of the actors in $V\setminus{}N$ is not taken into account and a negative initial decision is assumed. In non-oblivious influence models, as in \OLF, the initial decision of actors in $V\setminus{}N$ is taken into account under some considerations.

\begin{example}
In Figure~\ref{fig:evenpredecessors-ig} we provide an influence game $(G,f,3,\{2,3,4,5\})$ in which the models $\cM^o(\Gamma)$ and $\cM^n(\Gamma)$ do not coincide. Observe that, for $x=(0,1,1,0,0)$, $C_{\cM^o(\Gamma)}(x)=1$ but $C_{\cM^n(\Gamma)}(x)=0$.
\end{example}

\begin{figure}[t]
\centering
\begin{tikzpicture}[every node/.style={circle,scale=0.9}, >=latex]
\node[draw](a) at (1.5,0.0)[label=right:$1$] {2};
\node[draw](b) at (0.0,1.0)[label=above:$2$] {1};
\node[draw](c) at (1.0,1.0)[label=above:$3$] {1};
\node[draw](d) at (2.0,1.0)[label=above:$4$] {1};
\node[draw](e) at (3.0,1.0)[label=above:$5$] {1};
\draw[->] (b) to node {}(a);
\draw[->] (c) to node {}(a);
\draw[->] (d) to node {}(a);
\draw[->] (e) to node {}(a);
\end{tikzpicture}
\caption{An example of an influence graph.\label{fig:evenpredecessors-ig}}
\end{figure}
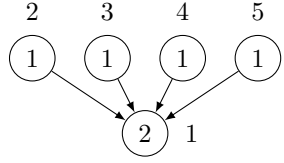

Now we analyze some basic properties of both oblivious and non-oblivious influence models. We start relating some particular decision-making models with simple games.

We say that a decision-making model $\cM$  is  \emph{monotonic}, on a set of actors $V$, if its collective decision function is monotonic with respect to inclusion, on $\cP(V)$. For a monotonic model, we define an associated  simple game $\Gamma_\cM$ in the natural way. A coalition $X$ is winning in $\Gamma_\cM$ if and only if $C_\cM(x(X))=1$. Our first results relates the \SAT measure and the \Rae index.

\begin{lemma}
\label{lem:Rae=Sat}
Let $\cM$ be a monotonic decision-making model   on  a set of actors $V$. For $i\in V$, $\SatM_\cM(i)=\RaeM_{\Gamma_\cM}(i)$.
\end{lemma}

\begin{proof}
Recall that, in  $\Gamma_\cM$, we have $\cW=\{X\subseteq V\mid C_\cM(x(X))=1\}$. Furthermore, for $i\in V$,
$\RaeM(i)=|\{X\in\cW \mid i\in X\}|+|\{X\notin\cW \mid i\notin X\}|$.
For $i\in{}V$ and $X\subseteq V$, we consider four cases.
\begin{itemize}
  \item If $X\in\cW$ and $ i\in X$, then $1=C_\cM(x(X))=x_i$.
  \item  If $X\in\cW$ and $ i\notin X$, then $1=C_\cM(x(X))\neq x_i=0$.
  \item  If $X\notin\cW$ and $ i\in X$, then $0=C_\cM(x(X))\neq x_i=1$.
  \item  If $X\notin\cW$ and $ i\notin X$, then $0=C_\cM(x(X))=x_i$.
\end{itemize}
\qqed\end{proof}

Next we prove that both the oblivious and non-oblivious models associated to an influence game are monotonic. 
\begin{lemma}
Let $\Gamma=(G,f,q,N)$ be an influence game. $\cM^o(\Gamma)$ and $\cM^n(\Gamma)$ are monotonic. 
\end{lemma}

\begin{proof}
{\em Oblivious model}. Let $\cM=\cM^o(\Gamma)$. For $X\subseteq{}X'\subseteq{}V$, $(X\cap{}N)\subseteq{}(X'\cap{}N)\subseteq{}N$. Thus, by the monotonicity of the spread of influence process  we know that $F(X)\subseteq F(X')$. Thus,
we have $C_\cM(X)\leq{}C_\cM(X')$.

{\em Non-oblivious model}. Let $\cM=\cM^n(\Gamma)$. For $X\subseteq{}V$ and $i\not\in{}X$, we consider two cases.
\begin{itemize}
  \item If $i\in{}N$, then $C_\cM(X)\leq{}C_\cM(X\cup{}\{i\})$ because of the monotonicity of $F$.
  \item If $i\not\in{}N$, then $(F(X\cap N)\cap P(i))\subseteq{}(F((X\cup\{i\})\cap N)\cap P(i))$. Thus, by the definition of the collective decision function,
we have $C_\cM(X)\leq{}C_\cM(X\cup{}\{i\})$.
\end{itemize}
\qqed\end{proof}

In order to relate \gOLF  with influence games we consider the following construction. Given a \gOLF $\cM=(G,r,q)$ we associate with $\cM$ the influence game $\Gamma(\cM)=(G,f,q,N)$ where $N=\tL\cup\tI$ and the labeling function $f$ is defined as
\[
f(i)=
\begin{cases}
 \Big\lceil r\delta^-(i)\Big\rceil & \mbox{if } i\in\tF\\
 1                                 & \mbox{if } i\in\tL\cup\tI.
\end{cases}
\]
Note that $N$ does not include the set of followers because followers never can enforce their personal conviction and their final decision depends exclusively on whether the opinion leaders can influence them or not. 

We denote the influence graph  $(G,f)$ of the associated influence game $\Gamma(\cM)$ as $G(\cM)$  and by $N(\cM)$ the corresponding set of players.

\begin{lemma}
\label{lem:OLF_as_IG}
Let $\cM$ be a \gOLF model and let $\cM\rq{}=\cM^n(\Gamma(\cM))$.
Then  the collective decision functions of  $\cM$ and $\cM\rq{}$ coincide.
\end{lemma}

\begin{proof}
Let $\cM\rq{}=\cM^n(\Gamma(\cM))$ be the non-oblivious influence model associated with $\cM=(G,r,q)$.
Let $X\subseteq{}V$ be a coalition so that $x(X)$ is the initial decision of the actors. Set $c=c^\cM(X)$ and $c'=c^{\cM\rq{}}(X)$ to be the corresponding final decision vectors.
Observe that, for $i\in\tL\cup\tI$, actor $i$ can not be influenced by any other actor in $V$; therefore, $i\in F(X\cap{}N)$ if and only if $i\in X$ and $c_i = c'_i$.
For $i\in{}\tF$, $\{j\in{}P(i)\mid x_j=1\} = F(X\cap{}N)\cap{}P(i)$. Therefore, $c_i = c'_i$ because the tie-breaking rule is the same in both models.
Thus, we have that $C_\cM = C_{\cM\rq{}}$.
\qqed\end{proof}

As a consequence of the previous result we have a way to map \gOLF models to a subfamily of the non-oblivious influence models. In general, a \gOLF cannot be cast as an oblivious influence model because the tie-breaking rules are different. Nevertheless, we can consider a subfamily in which ties do not arise.

\begin{definition}
An {\em \oOLF} is a \gOLF  $\cM=(G,r,q)$ in which $r=1/2$ and for all $i\in\tF$, $\delta^-(i)$ is odd.
\end{definition}

\begin{lemma}
\label{lem:oOLF}
Let $\cM$ be an \oOLF model and let $\cM\rq{}=\cM^o(\Gamma(\cM))$.
Then  the collective decision functions of  $\cM$ and $\cM\rq{}$ coincide.
\end{lemma}

\begin{proof}
Let $\cM=(G,1/2,q)$ be an \oOLF, let $\Gamma(\cM)$ be the influence game associated with $\cM$, and
let $X\subseteq{}V$ be a coalition so that $x(X)$ is the initial decision of the actors.
Let $\cM'=\cM^o(\Gamma(\cM))$. 
For $i\in\tL\cup\tI$, Lemma~\ref{lem:OLF_as_IG} shows that $c^\cM_i = c^{\cM'}_i$.
For $i\in{}\tF$, $\{j\in{}P(i)\mid x_j=1\} = F(X\cap{}N)\cap{}P(i)$.
Since $r=1/2$ and $\delta^-(i)$ is odd, $|F(X\cap{}N)\cap{}P(i)|\neq|P(i)|-|F(X\cap{}N)\cap{}P(i)|$.
Thus, the final decision vector of the collective decision function does not depend on the follower's initial decision.
Therefore, there are no ties arising in the predecessors' decision. So, the oblivious influence model verifies $C_{\cM'}=C_\cM$.
\qqed\end{proof}

Note that Lemma~\ref{lem:oOLF} is not true for a \gOLF where some follower has even in-degree. A counter example is the \OLF  $\cM=(G,1/2,3)$, whose graph $G$ and associated influence graph are depicted in Figure~\ref{fig:evenpredecessors}.

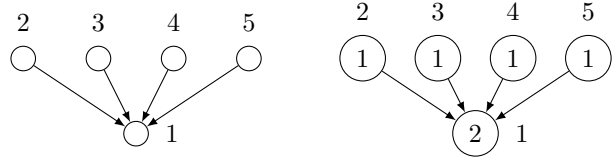
\begin{figure}[t]
\centering
\begin{tikzpicture}[every node/.style={circle,scale=0.9}, >=latex]
\node[draw](a) at (1.5,0.0)[label=right:$1$] {};
\node[draw](b) at (0.0,1.0)[label=above:$2$] {};
\node[draw](c) at (1.0,1.0)[label=above:$3$] {};
\node[draw](d) at (2.0,1.0)[label=above:$4$] {};
\node[draw](e) at (3.0,1.0)[label=above:$5$] {};
\draw[->] (b) to node {}(a);
\draw[->] (c) to node {}(a);
\draw[->] (d) to node {}(a);
\draw[->] (e) to node {}(a);
\end{tikzpicture}
\qquad
\begin{tikzpicture}[every node/.style={circle,scale=0.9}, >=latex]
\node[draw](a) at (1.5,0.0)[label=right:$1$] {2};
\node[draw](b) at (0.0,1.0)[label=above:$2$] {1};
\node[draw](c) at (1.0,1.0)[label=above:$3$] {1};
\node[draw](d) at (2.0,1.0)[label=above:$4$] {1};
\node[draw](e) at (3.0,1.0)[label=above:$5$] {1};
\draw[->] (b) to node {}(a);
\draw[->] (c) to node {}(a);
\draw[->] (d) to node {}(a);
\draw[->] (e) to node {}(a);
\end{tikzpicture}
\caption{A two layered bipartite graph $G$ and the influence graph $(G,f)$ corresponding to the \gOLF $(G,1/2,3)$.\label{fig:evenpredecessors}}
\end{figure}

\section{The hardness of computing the satisfaction measure}
\label{sec:SAT}

In this section we show that the \SAT problem is \#P-hard for \oOLF. In order
to do so we introduce some notation and define an auxiliary problem.

Let $(G,f)$ be an influence graph. For $i\in V(G)$, $\tF_i(G,f)$  denotes the set $\{j\in S_G(i)\mid|P_G(j)|=1\text{ and }f(j)=1\}$.
For $N\subseteq V(G)$ and $1\leq k\leq n$, $F_k(N,G,f)$ denotes the set $\{X\subseteq V(G)\mid |F(X\cap N)|=k\}$.
When there is no risk of ambiguity, we simply say $\tF_i$ or $F_k(N)$.
Note that $F_k(V)=\{X\subseteq V\mid |F(X)|=k\}$. Now we are able to define the auxiliary problem.
\begin{quote}
\begin{description}
\item{\EXPN}
\item{Instance:} 
An influence graph $(G,f)$, a set of vertices $N\subseteq V(G)$ and an integer $k$.
\item{Output:}   
$|F_k(N, G,f)|$.
\end{description}
\end{quote}
We sometimes consider that an instance of the \EXPN problem is an \oOLF $\cM=(G,r,q)$ by taking $G=G(\cM)$, $N=N(\cM)$ and some adequate value for $k$.

The following results show the relationship among the \SAT and the \EXPN problems for some oblivious influence models.

\begin{lemma}
\label{lem:SAT_for_oblivious}
Let $\Gamma = (G,f,q,N)$ be an influence game and let $\cM^o(\Gamma)$. For $i\in V(G)\setminus{}N$ or $i\in N$ with $f(i)=0$,  $\SatM_{\cM^o}(i)= 2^{n-1}$. 
\end{lemma}

\begin{proof}
Let $\cM= \cM^o(\Gamma)=(G,f,q,N)$,  let $Z=\{i\in N\mid f(i)=0\}$. 
Let $x$ be an initial decision vector and set $X=X(x) \cap (N\setminus Z)$.
Oserve that $F(X)=F(X\cup Z)$ and therefore the final decision is independent on the initial decision of those players in the set $Z$.

For $i\notin N$ or $i\in Z$,  
we provide the necessary and sufficient condition for
actor  $i$ being  satisfied. We consider two cases:
\begin{enumerate}
  \item If $X\in\cW$, $C_\cM(x)=1$, and player $i$ is satisfied only when $x_i=1$.
  \item If $X\notin\cW$, $C_\cM(x)=0$, and player $i$ is satisfied only when $x_i=0$.
\end{enumerate}
Therefore, for an initial decision vector of  the players in $N\setminus Z$, there is only one way, for player $i$, to complete it  in such a way that the collective decision coincides with player $i$'s decision. Thus we obtain $\SatM(i)=2^{n-1}$.
\qqed\end{proof}

In the following result we make use of a particular game construction. Let $(G,f,q,N)$ be an influence game. Define $R(G,f,i)$ as the influence graph $(G',f')$, where
$G'=G[V(G)\setminus(\tF_i \cup \{i\})]$, for $j\notin S_G(i)$, $f'(j)=f(j)$, and, for $j\in S_G(i)$, $f'(j)=\max\{f(j)-1,0\}$.

\begin{lemma}
\label{lem:SAT_for_oblivious2}
Let $\Gamma=(G,f,q,N)$ be an influence game. For $i\in N$ with $f(i)>0$ and $P_G(i)=\emptyset$, we have that 
$$\SatM_{\cM^o}(i) = 2^{n-1}+2^{n-|N|}\sum_{j=1}^{1+|\tF_i|} |F_{q-j}(N\setminus\{i\},R(G,f,i))|.$$
\end{lemma}

\begin{proof}
Let $\cM= \cM^o(\Gamma)=(G,f,q,N)$, let 
$x$ be an initial decision vector and set $X=X(x) \cap N$.
For $i\in N$ with $f(i)>0$ and $P_G(i)=\emptyset$,  we provide the conditions for $i$ to be satisfied. We consider three cases. 
\vspace{-\topsep}
\begin{enumerate}
  \item If $X\setminus\{i\}\in\cW$, then $C_\cM(x)=1$, so it must be $x_i=1$.
  \item If $X\setminus\{i\}\notin\cW$ and $X\notin\cW$, then $C_\cM(x)=0$, so it must be $x_i=0$.
  \item If $X\setminus\{i\}\notin\cW$ and $X\in\cW$, then $C_\cM(x)=1$, so it must be $x_i=1$.
\end{enumerate}
\vspace{-\topsep}
The first two cases provide a total of $2^{n-1}$ initial decision vectors for which the collective decision coincides with the initial decision of player $i$. To count the initial decision vectors for the third case, we consider the influence graph $R(G,f,i)$.
We have to separate from the rest those vertices in the set $\tF_i$ that can be influenced directly and only by $i$. Observe that all the vertices in $S_G(i)\setminus\tF_i$ have in-degree at least $2$. Now, for a coalition $Y$, it holds that $Y\in\cL$ and $Y\cup\{i\}\in\cW$ if and only if $Y\in F_{q-j}(N\setminus\{i\},R(G,f,i))$, for some $1\leq j\leq1+|\tF_i|$. Therefore, since the influence model is oblivious, there are $2^{n-|N|}\sum_{j=1}^{1+|\tF_i|}|F_{q-j}(N\setminus\{i\}, R(G,f,i))|$ additional initial decision vectors $z$ with $z_i=C_\cM(z)=1$.
\qqed\end{proof}

Note that, for $i\in N$ with $f(i)>0$, $P_G(i)=S_G(i)=\emptyset$ and   $\tF_i=\emptyset$,  
$$\SatM(i)=2^{n-1}+2^{n-|N|}|F_{q-1}(N\setminus\{i\},R(G,f,i))|.$$

Note also that, in \oOLF models, $N=\tL\cup\tI$, $V\setminus N=\tF$, and for $i\in N$, $P(i)=\emptyset$. Therefore, Lemmas~\ref{lem:SAT_for_oblivious} and~\ref{lem:SAT_for_oblivious2} provide the formulas for the satisfaction measure in those models. These results also show that, as expected, the opinion leaders have always a satisfaction greater than or equal to that of  the independent actors, and that both have always a satisfaction greater or equal than the followers.

The previous lemma does not provide a formula for the case in which the vertex $i\in N$ has $f(i)>0$ and $P_G(i)\neq\emptyset$. Although this case never occurs in \gOLF, it can be handled in other models by considering another graph construction.
Let $R_2(G,f,i)$ be the influence graph $(G'',f'')$, where $V(G'')=(V(G)\setminus\tF_i)\cup Z$, with $Z=\{z_1,\dots z_{2n}\}$ a set of new vertices, and $E(G'')$ is formed by the arcs in $G[V(G)\setminus\tF_i]$ and the set $\{(i,z_j)\mid 1\leq j\leq 2n\}$. The labeling function is given by $f''(j)=f(j)$, for $j\in V(G)\setminus S_G(i)$, by $f''(j)=\max\{f(j)-1,0\}$, for $j\in S_G(i)$, and by $f''(z_j)=1$, for $1\leq j\leq 2n$. The following result can be proved in the same way than Lemma~\ref{lem:SAT_for_oblivious2}, by replacing $R(G,f,i)$ by $R_2(G,f,i)$.

\begin{lemma}
\label{lem:SAT_for_oblivious3}
Let $(G,f,q,N)$ be an oblivious influence model, for $i\in N$ with $f(i)>0$ and $P_G(i)\neq\emptyset$,
$$\SatM(i)=2^{n-1}+2^{n-|N|}\sum_{j=1}^{1+|\tF_i|}|F_{q-j}(N\setminus\{i\},R_2(G,f,i))|$$
\end{lemma}

For our hardness result we consider a variation of the {\em counting vertex cover problem}~\cite{GJ79}:
\begin{quote}
\begin{description}
\item{\VC}
\item{Instance:} 
An undirected graph $G=(V,E)$.
\item{Output:}  
Number of vertex covers of $G$ with size $\frac{2}{3}|V|$, i.e.,\\
$\big|\big\{X\subseteq V\mid\text{ for }\{i,j\}\in E, \ \ \{i,j\}\cap X \neq \emptyset \text{ and } |X|=\frac{2}{3}|V|\big\}\big|$.
\end{description}
\end{quote}

It is known that the problem of computing in a graph the number of independent sets with size exactly $\frac13|V|$ is \#P-hard~\cite{Hof10}.
Hence, as the complement of an independent set is a vertex cover, the same result shows that \VC is \#P-hard.

\begin{figure}[t]
\centering
\begin{tikzpicture}[every node/.style={circle,scale=0.9}, >=latex]
\node(x) at (1.0,3.0)[label=above:\mbox{{\large$G$}}] {};
\node[draw](a1) at (0.0,2.5)[label=above:$1$] {};
\node[draw](a2) at (1.0,2.5)[label=above:$2$] {};
\node[draw](a3) at (2.0,2.5)[label=above:$3$] {};
\draw[-] (a1) to node[below] {$e_1$}(a2);
\draw[-] (a2) to node[below] {$e_2$}(a3);
\node(x) at (6.25,2.5)[label=above:\mbox{\large{$(G_1,f_1)$}}] {};
\node[draw](a) at (4.5,2.5)[label=above:$1$] {};
\node[draw](b) at (5.5,2.5)[label=above:$2$] {};
\node[draw](c) at (6.5,2.5)[label=above:$3$] {};
\node[draw](z) at (8.0,2.5)[label=above:$z$] {};
\node[draw](e11) at (3.0,0.5)[label=below:$e^1_1$] {};
\node[draw](e12) at (3.5,0.5)[label=below:$e^1_2$] {};
\node[draw](e21) at (4.5,0.5)[label=below:$e^2_1$] {};
\node[draw](e22) at (5.0,0.5)[label=below:$e^2_2$] {};
\node[draw](e31) at (6.0,0.5)[label=below:$e^3_1$] {};
\node[draw](e32) at (6.5,0.5)[label=below:$e^3_2$] {};
\node[draw](e41) at (7.5,0.5)[label=below:$e^4_1$] {};
\node[draw](e42) at (8.0,0.5)[label=below:$e^4_2$] {};
\node[draw](e51) at (9.0,0.5)[label=below:$e^5_1$] {};
\node[draw](e52) at (9.5,0.5)[label=below:$e^5_2$] {};
\node(x) at (3.25,0)[label=below:$E_1$] {};
\node(x) at (4.75,0)[label=below:$E_2$] {};
\node(x) at (6.25,0)[label=below:$E_3$] {};
\node(x) at (7.75,0)[label=below:$E_4$] {};
\node(x) at (9.25,0)[label=below:$E_5$] {};
\draw[->] (a) to node {}(e11);
\draw[->] (a) to node {}(e21);
\draw[->] (a) to node {}(e31);
\draw[->] (a) to node {}(e41);
\draw[->] (a) to node {}(e51);
\draw[->] (b) to node {}(e11);
\draw[->] (b) to node {}(e12);
\draw[->] (b) to node {}(e21);
\draw[->] (b) to node {}(e22);
\draw[->] (b) to node {}(e31);
\draw[->] (b) to node {}(e32);
\draw[->] (b) to node {}(e41);
\draw[->] (b) to node {}(e42);
\draw[->] (b) to node {}(e51);
\draw[->] (b) to node {}(e52);
\draw[->] (c) to node {}(e12);
\draw[->] (c) to node {}(e22);
\draw[->] (c) to node {}(e32);
\draw[->] (c) to node {}(e42);
\draw[->] (c) to node {}(e52);
\draw[->] (z) to node {}(e11);
\draw[->] (z) to node {}(e12);
\draw[->] (z) to node {}(e21);
\draw[->] (z) to node {}(e22);
\draw[->] (z) to node {}(e31);
\draw[->] (z) to node {}(e32);
\draw[->] (z) to node {}(e41);
\draw[->] (z) to node {}(e42);
\draw[->] (z) to node {}(e51);
\draw[->] (z) to node {}(e52);
\end{tikzpicture}
\caption{Influence graph $(G_1,f_1)$ obtained from an undirected graph $G$.\label{fig_demo_Th1}}
\end{figure}
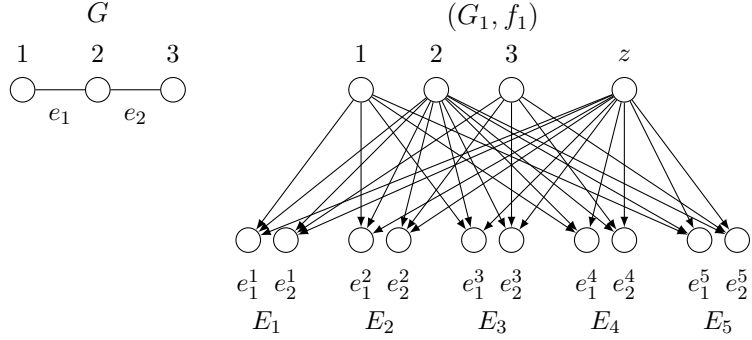

\begin{theorem}
\label{the:Fnk_SAT_oOLF}
The \EXPN problem for \oOLF is \#P-hard.
\end{theorem}

\begin{proof}
We provide a reduction from the \VC problem.
Let $G=(V,E)$ be an instance of \VC. Without loss of generality we assume that $G$ is connected, $n=|V|$ is a multiple of $3$ and $n\geq 6$. Let $m=|E(G)|$ and $E=\{e_1,\dots, e_m\}$.

We construct a two layered bipartite graph $G_1$ associated to $G$ as follows (see Figure~\ref{fig_demo_Th1} for an example). The set of vertices is $V(G_1)=V\cup E_1\cup E_2\cup\dots\cup E_{n+2}\cup\{z\}$, where $z$ is a new vertex and $E_j=\{e_1^j,\dots, e_m^j\}$, $1\leq j\leq n+2$, is formed by vertices associated to the edges of $E$. Observe that $n_1=|V(G_1)|=n+(n+2)m+1$. The set of arcs is the following:
\begin{align*}
E(G_1) =  \{(u,e_k^j) &\mid u\in V, 1\leq j\leq n+2,  1\leq k\leq m \text{ and }  u\in e_k\} \\
          &\cup \{(z,a) \mid a\in E_j, 1\leq j\leq n+2\}.
\end{align*}
Note that in $G_1$  all the vertices have in-degree either $0$ or $3$.
For the labeling function $f_1$ associated to the influence graph of the game  $\Gamma(G_1,1/2,n_1)$ we have $f_1(u)=1$, for $u\in V$, $f_1(z)=1$, and $f_1(u)=2$, for $u\notin(V\cup\{z\})$.

Now we define the reduction from \VC to \EXPN, which associates to $G$ the following instance $h(G)$ for the \EXPN problem:
$$h(G)= \left((G_1,f_1), V\cup\{z\},\frac{2}{3}n+(n+2)m+1\right).$$

Let $X\subseteq V$ and let $\alpha=|X|$. We analyze the expansion of the sets $X\cup\{z\}$ and $X$ in the influence graph $(G_1,f_1)$.

When the initial activation set is $X\cup\{z\}$, we have two cases, either $X$ is a vertex cover or not. When $X$ is a vertex cover, all vertices corresponding to edges get activated, so $|F(X\cup\{z\})|=\alpha + (n+2)m+1$. This last quantity is equal to the required size only when $\alpha=\frac{2}{3}n$. When $X$ is not a vertex cover, then $\alpha\leq n-2$ and at least one edge $e\in E$ is not covered. In consequence,   $F(X\cup\{z\})$ can not influence all the vertices in $E_1\cup\ldots\cup E_{n+2}$. Therefore, $|F(X\cup\{z\})|\leq \alpha+(n+2)(m-1)+2\leq n-2 + (n+2)m - (n+2) + 2 = (n+2)m - 2$ which is strictly smaller than the required size.

Now consider the case when the initial activation set is $X$. Recall that $G$ is connected. If for $\{u,v\}\in E$ it holds that $\{u,v\}\subseteq X$, then we have $|F(X\cup\{z\})| = n+(n+2)m$, which is greater than the required size. Otherwise, $F(X)$ can not influence all the vertices in $E_1\cup\ldots\cup E_{n+2}$, we have $\alpha\leq n-1$, and hence $|F(X)|\leq\alpha + (n+2)(m-1)\leq n-1 + (n+2) m - (n+2) = (n+2) m - 3$ which is strictly smaller than the required size.

From the previous case analysis, we have that the elements in $F_k(V\cup\{z\})$, for $(G_1,f_1)$, are in a one-to-one correspondence with the vertex covers of size $\frac{2}{3}n$ in $G$. As the reduction can be computed trivially in polynomial time, the claim holds.
\qqed\end{proof}

The hardness of the \EXPN problem does not rule out the possibility of having some particular cases for which the \EXPN problem is computationally easy. For example, when $f$ is strictly positive and the parameter $k$ of the problem is smaller than the minimum label over the actors not in $N$, i.e., $k<\min\{f(i)\mid i\in V\setminus N\}$, it is easy to see that, for an oblivious influence model, $|F_k(N)|=\binom{n-|N|}{k}$.

Now we can combine the previous results of this section to provide a hardness proof for the \SAT problem.

\begin{theorem}
\label{cor:Fnk_SAT}
The \SAT problem for \oOLF is \#P-hard.
\end{theorem}

\begin{proof}
We prove hardness by showing a polynomial time reduction from the \EXPN problem on bipartite influence graphs to the \SAT problem. Consider an instance of \EXPN given by a bipartite influence graph $(G,f)$, a set $N\subseteq V(G)$ and an integer~$k$.  Let $(G',f')$ be the influence graph obtained from $(G,f)$ by adding an isolated vertex $z$ with label~$1$. Finally, we consider the influence game $\Gamma(G,2/3,k+1)=(G',f',N\cup \{z\},k+1)$ and the instance $(\Gamma(G,f'),z)$ of the \SAT problem.

In order to compute $\SatM(z)$ in $\Gamma(G,f)$, according to Lemma~\ref{lem:SAT_for_oblivious2} we have to consider the reduced influence graph $R(G',f',z)$ and the parameter $q=k+1$. By construction $R(G',f',z)=(G,f)$ and thus we have $$\SatM(z)=2^n+ 2^{n+1-|N|} |F_k(N,(G,f))|.$$ Therefore, if we could solve the \SAT problem in polynomial time, we would also be able to solve \EXPN in polynomial time.
\qqed\end{proof}

As a consequence of the previous result and Lemma~\ref{lem:oOLF} we have the following.

\begin{corollary}
The \SAT problem for \gOLF and the oblivious and non-oblivious influence models is \#P-hard.
\end{corollary}

As a consequence of the previous result and Lemma~\ref{lem:Rae=Sat}, the problems of computing the Rae index and the Banzhaf value are \#P-hard for the families of simple games associated to oblivious and non-oblivious influence models, \oOLF and \gOLF.

\section{Strong hierarchical influence models}
\label{sec:SATpol}

In this section we focus our attention on one particular topology of the influence graph that provide oblivious and non-oblivious influence models where the \SAT problem is polynomial time solvable. In the previous section we have shown that  the \SAT problem is hard for two layered bipartite graphs.  Thus,  the  potentially  tractable subfamilies of influence graphs must overcome this difficulty by considering stronger or simpler topologies. Our first subfamilies are based on multilayered bipartite graphs with strong restrictions on the connections among the layers. In a multilayered bipartite graph, there are some intermediate actors that we called {\em mediators}, who act as intermediate layers of influence expansion between opinion leaders and followers.

\begin{figure}[t]
\centering
\begin{tikzpicture}[every node/.style={circle,scale=0.8}, >=latex]
\node[draw](a) at (0.5,3.0)[label=above:$1$] {1};
\node[draw](b) at (1.5,3.0)[label=above:$2$] {1};
\node[draw](c) at (3.5,3.0)[label=above:$3$] {1};
\node[draw](d) at (0.0,2.0)[label=below:$4$] {2};
\node[draw](e) at (1.0,2.0)[label=below:$5$] {2};
\node[draw](f) at (2.0,2.0)[label=below:$6$] {2};
\node[draw](g) at (3.0,2.0)[label=below:$7$] {3};
\node[draw](h) at (4.0,2.0)[label=below:$8$] {3};
\node[draw](i) at (1.5,1.0)[label=below:$9$] {3};
\node[draw](j) at (2.5,1.0)[label=below:$10$]{4};
\node[draw](k) at (2.0,0.0)[label=below:$11$]{2};
\node[draw](l) at (5.5,3.0)[label=above:$12$]{1};
\node[draw](m) at (5.0,2.0)[label=below:$13$]{1};
\node[draw](n) at (6.0,2.0)[label=below:$14$]{1};
\node[draw](o) at (7.0,3.0)[label=above:$15$]{1};
\node[draw](p) at (8.0,3.0)[label=above:$16$]{1};
\draw[->] (a) to node {}(d);
\draw[->] (a) to node {}(e);
\draw[->] (a) to node {}(f);
\draw[->] (b) to node {}(d);
\draw[->] (b) to node {}(e);
\draw[->] (b) to node {}(f);
\draw[->] (c) to node {}(g);
\draw[->] (c) to node {}(h);
\draw[->] (d) to node {}(i);
\draw[->] (d) to node {}(j);
\draw[->] (e) to node {}(i);
\draw[->] (e) to node {}(j);
\draw[->] (f) to node {}(i);
\draw[->] (f) to node {}(j);
\draw[->] (g) to node {}(i);
\draw[->] (g) to node {}(j);
\draw[->] (h) to node {}(i);
\draw[->] (h) to node {}(j);
\draw[->] (i) to node {}(k);
\draw[->] (j) to node {}(k);
\draw[->] (l) to node {}(m);
\draw[->] (l) to node {}(n);
\end{tikzpicture}
\caption{A strong hierarchical influence graph.\label{fig:oOLFM_t4}}
\end{figure}
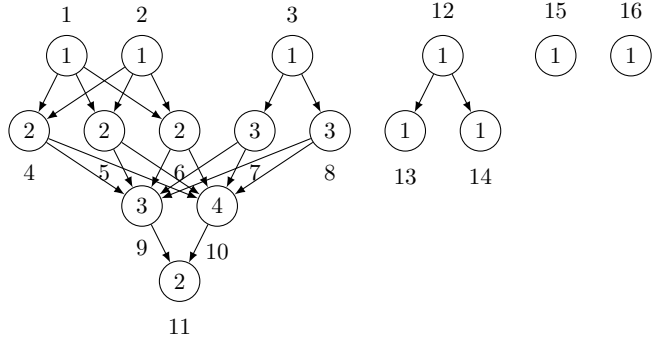

The family of graphs is defined recursively using two graph operations. 
Given two graphs $H_1$ and $H_2$ with $V(H_1)\cap V(H_2)=\emptyset$ their \emph{disjoint union} is the graph  $H_1+H_2 =(V(H_1)\cup V(H_2), E(H_1)\cup E(H_2))$. 
Given a graph $H$,  the \emph{one layer extension to a set  $V'\neq \emptyset$  of new vertices} ($V(H)\cap V'=\emptyset$) is   the graph $H\otimes V'$ constructed as follows
\begin{align*}
V(H \otimes V')= & V(H)\cup V' \mbox{ and}\\
E(H \otimes V')= & E(H)\cup \{(u,v)\mid u\in\tFI(H), v\in V'\}.
\end{align*}

Observe that we have  $\tL(H_1+H_2)=\tL(H_1) \cup \tL(H_2)$, $\tI(H_1+H_2)=\tI(H_1)\cup\tI(H_2)$, $\tF(H_1+H_2)=\tF(H_1)\cup \tF(H_2)$ and $\tFI(H_1+H_2)=\tFI(H_1) \cup \tFI(H_2)$. 
Furthermore,  $\tL(H \otimes V')=\tL(H)\cup \tI(H)$, $\tI(H \otimes V')=\emptyset$, and $\tF(H \otimes V')=\tFI(H \otimes V')=V'$. 

As base case we use graphs with only isolated vertices. The  family is completed  by taking the closure under the two graph  operations defined above.

\begin{definition}
\label{def:SHIM}
The family of {\em strong hierarchical graphs} is defined recursively as follows.
\begin{itemize}
\item The graph $I_a$, for $a>0$, is a strong hierarchical graph.
\item If $H_1$ and $H_2$ are disjoint strong hierarchical graphs, the graph $H_1 + H_2$ is a strong hierarchical graph.
\item If $H$ is a strong hierarchical graph and   $V'\neq \emptyset$  is a set of vertices with $V(H)\cap V' =\emptyset$,  the graph $H\otimes V'$ 
is a strong hierarchical graph.
\end{itemize}
A {\em strong hierarchical influence graph} is an influence graph $(G,f)$ where $G$ is a strong hierarchical graph.
A {\em strong hierarchical influence game} is an influence game $(G,f,q,N)$ where $G$ is a strong hierarchical graph and $N=\tL(G)\cup\tI(G)$.
\end{definition}

In Figure~\ref{fig:oOLFM_t4} we provide an example of a strong hierarchical influence graph. Observe that the vertices of a strong hierarchical graph can be partitioned into layers so that edges occur only between vertices of consecutive layers. Furthermore, by removing the vertices with out-degree zero in a connected strong hierarchical graph, we obtain a decomposition formed by connected strong hierarchical graphs and possibly a  set of independent vertices. By repeatedly applying this process we can obtain a decomposition allowing to reconstruct the graph from several independent sets using disjoint union and adequate layer extensions. The graph $G$ given in  Figure~\ref{fig:oOLFM_t4} can be decomposed  as
\[[([(H_1 \otimes \{4,5,6\}) + (H_2 \otimes \{7,8\})] \otimes \{9,10\}) \otimes \{11\}] + [ (H_3 \otimes \{13,14\}) + H_4]\]
where $H_1=(\{1,2\},\emptyset)$, $H_2=(\{3\},\emptyset)$, $H_3=(\{12\},\emptyset)$ and $H_4=(\{15,16\},\emptyset)$.

We start providing an algorithm to solve the \EXPN problem for strong hierarchical influence games. Our algorithm uses dynamic programming and uses the previous decomposition to guide the computation of some adequate tabulated values. 

\begin{lemma}
\label{lem:complete_t2}
Let $(G,f,q,N)$ be a strong hierarchical influence game. For $1\leq{}k\leq n$,
the values $|F_k(N,G,f)|$ can be computed in polynomial time.
\end{lemma}

\begin{proof}
Let $n=|V(G)|$, for $0\leq b\leq a\leq n$ and $0\leq b\leq|\tFI(G)|$, consider the following values $T(a,b)$:
$$T(a,b)=|\{X\subseteq N\mid |F(X)|=a \text{ and } |F(X)\cap \tFI(G)|=b\}|.$$
Observe that, if we can compute in polynomial time an array holding all the $T(a,b)$ values, then we can obtain $|F_k(N)|$ in polynomial time as
$$|F_k(N)| = 2^{n-|N|}\sum_{0\leq b\leq |\tFI(G)|}T(k,b).$$

Let us show, by induction on the structure of the graph $G$, how  an array storing the desired values can be obtained from the corresponding arrays for adequate subgraphs.
According to  Definition~\ref{def:SHIM}, the base case are sets of isolated vertices.

\smallskip
\noindent
\textbf{Base case:}  $H=I_\alpha$.

In this case  all the actors are independent, so $|\tFI(H)|=\alpha$ and, for $X\subseteq V$, $F(X)=X\cup \{i\in V\mid f(i)=0\}$. Therefore, for $0\leq b\leq a\leq \alpha$,
we have 
$$T(a,b)=
\begin{cases}
\binom{\alpha-\gamma}{b-\gamma}\ 2^\gamma & \text{if  $a=b > \gamma$}\\
0 & \text{otherwise}
\end{cases}$$
where $\gamma=|\{i\in V\mid f(i)=0\}|$.
All those values can be computed in polynomial time.

The inductive step of the proof is divided into two cases.

\smallskip
\noindent
\textbf{Case 1:}  $H= H_1 + H_2$,  for some strong hierarchical graphs $H_1$, $H_2$, recall that we also have that  $V(H_1)\cap V(H_2) = \emptyset$.

Let $T_1, T_2$ be the tables corresponding to the graphs $H_1$ and $H_2$, respectively, and let $n_1=|V(H_1)|$ and $n_2=|V(H_2)|$. As the graphs are disjoint, we have the following expression, for $0\leq b \leq a \leq n$ and $0\leq b\leq |\tFI(H)|$:
\[
T(a,b) = \hskip -20mm
\sum_{\begin{array}{c}
a_1+a_2=a\\
b_1+b_2=b\\
0\leq b_i \leq a_i \leq n_i \text{ and } 0\leq b_i \leq |\tFI(H_i)|, 1\leq i\leq 2
\end{array}
} \hskip -20mm T_1(a_1,b_1)  T_2(a_2,b_2).
\]
Those values can be computed in polynomial time using a multiple  scanning as follows:

\begin{tabbing}
123456789\=1234\=1234\=\kill
\> Initialize $T(a,b)$ to 0\\
\> for $a_1 = 0$ to $n_1$, $a_2 = 0$ to $n_2$,  $b_1 = 0$ to $a_1$,  $b_2 = 0$ to $a_2$\\
\>\>\> $T(a_1+a_2, b_1+b_2) = T(a_1+a_2, b_1+b_2) + T_1(a_1, b_1) T_2(a_2,b_2)$.
\end{tabbing}

\smallskip
\noindent
\textbf{Case 2:} $H = H' \otimes V'$  for some strong hierarchical graph $H\rq{}$ and a set $V'\neq \emptyset$ with $V(H')\cap V'=\emptyset$. 

In such a case $\tF(H)=V'$ and the graph $H'$ is a strong hierarchical graph with one layer less. Let $T'$ be
the array corresponding to the graph $H'$.  Let $n'=|V(H')|$, $\beta'=|\tFI(H')|$ and $\beta=|\tFI(H)|$. Recall that, in the construction of $H$, the arcs added to   $H'$ connect in an all-to-all fashion the vertices in $\tFI(H')$ with the vertices in $\tF(H)$. 

To express the values of $T$ we use an auxiliary table $R(c)$, $0\leq c\leq n$, defined as
$$R(c) =|\{v\in \tF(H) \mid f(v) \leq c\}|.$$
A vector storing the values of $R$ can be computed in polynomial time by sorting the set $\tF(H)$ in increasing order of labels and performing a scanning of the sorted table. 

We get the following expression, for $0\leq b\leq a \leq n$ and $0\leq b\leq \beta$:
\[
T(a,b)= \hskip -8mm
\sum_{\begin{array}{c}
a' + R(b')=a\\
R(b')=b\\
0\leq b'\leq a' \leq n' , 0\leq b' \leq \beta'
\end{array}} \hskip -10mm T'(a',b').
\]
Those values can be computed in polynomial time using a double scanning as follows:
\begin{tabbing}
123456789\=1234\=1234\=\kill
\> Initialize $T(a,b)$ to 0\\
\> for $0\leq b'\leq a' \leq n'$ , $0\leq b' \leq \beta'$\\
\>\>\> $T(a' + R(b'),R(b'))=T(a' + R(b'),R(b')) + T'(a',b')$.
\end{tabbing}

Note that given a graph, it is possible to know whether it is a strong hierarchical graph in polynomial time.
Moreover, given a strong hierarchical graph it is possible to find a decomposition, according to the definition,  as described above in polynomial time in order
to guide the application of the computation rules in polynomial time. 
\qqed\end{proof}

It is easy to see that, for a strong hierarchical influence graph $(G,f)$ and an actor $i\in \tL(G)\cup \tF(G)$, the graph $R(G,f,i)$, constructed as in Lemma~\ref{lem:SAT_for_oblivious2}, is a strong hierarchical influence graph. Therefore, we can use Lemma~\ref{lem:SAT_for_oblivious2} together with the previous algorithm to compute $\SatM(i)$ in polynomial time. Thus, we have the following.

\begin{theorem}\label{theo:SAT-Oblivious}
The \SAT problem, for  oblivious models corresponding to strong hierarchical influence games, is polynomial time solvable.
\end{theorem}

Finally we extend the previous computation to the non-oblivious models. In this case we provide a recursive algorithm that allow us to compute directly the satisfaction measure.

\begin{theorem}\label{theo:SAT-non-Oblivious}
The \SAT problem, for  non-oblivious models corresponding to strong hierarchical influence games, is polynomial time solvable.
\end{theorem}

\begin{proof}
Let $\cM=(G,f,q,N)$ be an oblivious influence model.  Assume that $G$ is a strong hierarchical graph.  We first compute a decomposition of $G$ according to the recursive definition. Recall that such a decomposition can be obtained in polynomial time. 

Given a vertex $u$, our algorithm to compute  $\SatM(u)$ computes some tables with partial results, one  for each of the subgraphs in the decomposition of $G$. Finally,  the algorithm  combines the tables corresponding to $G$ to get  $\SatM(u)$. 
Observe, the vertex $u$ is present only in a subset of the graphs appearing in the decomposition of $G$. We have to compute  different kind of information for a subgraph $H$ when $u\in V(H)$ than when $v\notin H$. When $u\in V(H)$ we have to keep track of $u$\rq{}s initial decision. 
  
Let $n=|V(H)|$ and $\cM=\cM^n(H)$. When $u\notin V(H)$, we consider the values $S(a,b)$, for $0\leq a\leq n$ and $0\leq b\leq|\tFI(H)|$, defined as follows:
$$S(a,b)=|\{x\in\{0,1\}^n\mid |\{i\mid c_i^\cM(x) = 1\}|=a \text{ and } |F(X(x)\cap N)\cap \tFI(H)|=b\}|.$$
When $u\in V(H)$, we consider two sets of values  $S_0(a,b)$ and $S_1(a,b)$,  for $0\leq a\leq n$ and $0\leq b\leq|\tFI(H)|$, defined as
\begin{align*}
S_0(a,b) &=|\{x\in\{0,1\}^n\mid x_u=0 \\
& \qquad \text{ and } |\{i\mid c_i^\cM(x) = 1\}|=a\\
& \qquad \text{ and } |F(X(x)\cap N)\cap \tFI(G)|=b\}|,\\
S_1(a,b) &=|\{x\in\{0,1\}^n\mid x_u=1 \\
& \qquad \text{ and } |\{i\mid c_i^\cM(x) = 1\}|=a \\
& \qquad \text{ and } |F(X(x)\cap N)\cap \tFI(G)|=b\}|.
\end{align*}

Observe that, if we can compute in polynomial time arrays holding all the $S_0(a,b)$ and $S_1(a,b)$ values,  for the graph $G$ (which indeed contains $u$),  we can express  $\SatM(u)$ as
$$\SatM(u)=  \sum_{0\leq a < q} \, \sum_{0\leq b\leq \tFI(G)} S_0(a,b) + \sum_{q\leq a\leq n} \, \sum_{0\leq b\leq \tFI(G)} S_1(a,b) .$$
Thus $\SatM(u)$ could be  computed in polynomial time.

Let us show, by induction on the structure of the graph $G$, how  an array storing the desired values can be obtained from the corresponding arrays for adequate subgraphs. Recall that according to the definition the labeling function is strictly positive.
According to  Definition~\ref{def:SHIM}, the base case is a set of isolated vertices.

\smallskip
\noindent
\textbf{Base case:}  $H=I_\alpha$, for some $\alpha>0$.

As all the actors are independent, $|\tFI(H)|=\alpha$ and, for $X\subseteq V$, $F(X)=X$. Furthermore, $c_i^\cM=1$ if and only if  $i\in F(X(x))$.  The expressions are similar to those of the base case for the expansion problem  when $\gamma =0$.
Therefore, for $0\leq a\leq \alpha$ and $0\leq b\leq \alpha$, when $u\notin V(H)$, we have  the expression
$$S(a,b)=
\begin{cases}
\binom{\alpha}{b}& \text{if  $a=b$,} \\
0 & \text{otherwise.}
\end{cases}$$
When $u\in V(H)$, we have to derive expressions for the two cases. Observe that, if $x_u=0$, $u$ does not form part of the initial $X$, so we can select vertices from $V\setminus\{u\}$. If $x_u=1$, $u$ must be part of any $X$, so we have to select one  vertex less. Therefore, we have
$$S_0(a,b)=
\begin{cases}
\binom{\alpha-1}{b}& \text{if  $a=b$,}\\
0 & \text{otherwise,}
\end{cases}$$
and 
$$S_1(a,b)=
\begin{cases}
\binom{\alpha-1}{b-1}& \text{if  $a=b$,}\\
0 & \text{otherwise.}
\end{cases}$$
Note that all those values can be computed in polynomial time.

\smallskip
\noindent
\textbf{Case 1:}  $H=H_1 + H_2$, for some disjoint strong hierarchical influence graphs, therefore   $V(H_1)\cap V(H_2)=\emptyset$.

For $i\in\{1,2\}$, let $n_i=|V(H_i)|$, $\beta_i=\tFI(H_i)$, and let $S^i, S^i_0,  S^i_1$  be the tables corresponding to the graph $H_i$. Finally, set  $\beta=\beta_1+\beta_2$. As the graphs are disjoint,for $0\leq b \leq a \leq n$ and $0\leq b\leq \beta$, when $u\notin V(H)$,  we have the following expression
\[
S(a,b) = \hskip -20mm
\sum_{\begin{array}{c}
a_1+a_2=a\\
b_1+b_2=b\\
0\leq a_i \leq n_i \text{ and } 0\leq b_i \leq \beta_i, 1\leq i\leq 2
\end{array}
} \hskip -20mm S^1(a_1,b_1)  \, S^2(a_2,b_2).
\]
When $u\in V(H)$ we assume w.l.o.g. that $u\in V(G_1)$ and we have
\[
S_0(a,b) = \hskip -20mm
\sum_{\begin{array}{c}
a_1+a_2=a\\
b_1+b_2=b\\
0\leq a_i \leq n_i \text{ and } 0\leq b_i \leq \beta_i, 1\leq i\leq 2
\end{array}
} \hskip -20mm S_0^1(a_1,b_1) \,  S^2(a_2,b_2) \]
and 
\[S_1(a,b) = \hskip -20mm
\sum_{\begin{array}{c}
a_1+a_2=a\\
b_1+b_2=b\\
0\leq a_i \leq n_i \text{ and } 0\leq b_i \leq \beta_i, 1\leq i\leq 2
\end{array}
} \hskip -20mm S_1^1(a_1,b_1)  \, S^2(a_2,b_2).
\]

Those values can be computed in polynomial time using a multiple  scanning similar to the one used in  the proof of Lemma~\ref{lem:complete_t2}.

\smallskip
\noindent
\textbf{Case 2:} $H = H' \otimes V'$, for some  strong hierarchical graph $H$ and a set $V'\neq\emptyset$ and $V'\cap V(H')=\emptyset$. 

In such a case $\tF(G)= V'$. Let $S'$, $S_1'$ and $S_2'$ be
the arrays corresponding to $H'$.  Let $n'=|V(H')|$, $\alpha=|\tFI(H')|$ and $\beta=|\tFI(H)|$. 
Recall that, in the construction of $G$, the arcs added to $G'$ connect in an all-to-all fashion the vertices in $\tFI(H')$ with the vertices in $\tF(H)$. 

To express the values of  $S$, $S_0$ and $S_1$ we use, as before,  an auxiliary table $R(c)$, $0\leq c\leq \alpha$, defined as
$$R(c) =|\{v\in \tF(G) \mid f(v) \leq c\}|$$
which can be computed in polynomial time. Note that $R(c)$ accounts for the number of actors in the added layer when $c$ followers of $G'$ are influenced.  

We need also information for other relevant sets.  For $0\leq c\leq \alpha$, define
\begin{align*}
A_1(c)&= \{v\in \tF(H) \mid f(v) \leq c \text{ and } \alpha-c < f(v)\},\\
A_2(c)&= \{v\in \tF(H) \mid f(v) \leq \alpha - c \text{ and } c < f(v)\},\\
A_3(c) & = \tF(H) - A_1(c) - A_2(c).
\end{align*}
Finally set $R_1(c)=|A_1(c)|$, $R_2(c)=|A_2(c)|$ and $R_3(c)=|A_3(c)|$.
All those sets and values can be precomputed, for any possible value of $c$, in polynomial time.

As the connection to the final layer is complete and $f$ is positive, for a set $X\subseteq\tFI(H')$,  we have  $F(X)=X \cup\{u\in\tF(H)\mid|X|\geq f(u)\}$.
Using this information, we know that a subset of opinion leaders $X\subseteq\tL(H)$  with $\gamma=|F(X)\cap \tFI(H')|$ will expand its influence to all the followers $i$ for which $f(i)\leq \gamma$.

For an initial decision vector $x\in\{0,1\}^{n}$, let $\gamma(x)=|F(X(x)\cap \tL(H))\cap \tFI(H')|$.
Observe that the associated final decision vector,  for $i\in\tF(G)$, can be expressed as
\[c_i^{\cM}(x) = 
\begin{cases}
 1  & \mbox{if } \gamma(x) \geq f(i) \text{ and } \alpha-\gamma(x) < f(i)\\
 0  & \mbox{if } \alpha-\gamma(x) \geq f(i) \text{ and } \gamma(x) < f(i)\\
x_i & \mbox{otherwise.}\\
\end{cases}
\, = \, 
\begin{cases}
 1  & \mbox{if } i\in A_1(\gamma(x)),\\
 0  & \mbox{if } i\in A_2(\gamma(x)),\\
x_i & \mbox{if } i\in A_3(\gamma(x)).\\
\end{cases}
\]
Taking into account the last expression, we can count those initial decision vectors giving raise to the prescribed number of 1's in the final decision vector.  We get the following expressions.

\smallskip
When $u\notin V(H)$, we have that, for $0\leq a\leq n$ and $0\leq b\leq \beta$, 
$$S(a,b)= \hspace{-1cm} \sum_{\begin{array}{c}
a'+  R_1(b') + \delta=a\\
R(b')=b\\
0\leq a'\leq n', \, 
0\leq b '\leq \alpha\\
0\leq \delta\leq R_3(b')
\end{array}} \hspace{-1cm} S'(a',b') \,   \binom{R_3(b')}{\delta} \, 2^{R_1(b')+R_2(b')}.$$

When $u\in V(H)$ we have two cases.

If $u\in V(H')$,  for $0\leq a\leq n$ and $0\leq b\leq \beta$,
$$S_0(a,b)= \hspace{-1cm}\sum_{\begin{array}{c}
a'+  R_1(b') + \delta=a\\
R(b')=b\\
0\leq a'\leq n', \, 
0\leq b'\leq \alpha\\
0\leq \delta\leq R_3(b')
\end{array}}\hspace{-1cm}S_0'(a',b') \,  \binom{R_3(b')}{\delta} \, 2^{R_1(b')+R_2(b')}$$
and
$$S_1(a,b)= \hspace{-1cm}\sum_{\begin{array}{c}
a'+  R_1(b') + \delta=a\\
R(b')=b\\
0\leq a'\leq n', \, 
0\leq b'\leq \alpha\\
0\leq \delta\leq R_3(b')
\end{array}}\hspace{-1cm}S_1'(a',b') \,  \binom{R_3(b')}{\delta} \, 2^{R_1(b')+R_2(b')}.$$

If $u\in V'=\tF(H)$, for $0\leq a\leq n$ and $0\leq b\leq \beta$, 
\begin{align*}
S_0(a,b) &= \hspace{-1cm}\sum_{\begin{array}{c}
a'+  R_1(b') + \delta=a\\
R(b')=b\\
0 \leq a'\leq n', \,
0\leq b'\leq \alpha\\
0\leq \delta\leq R_3(b')\\
u\notin A_3(b')
\end{array}} \hspace{-1cm} S'(a',b') \,  \binom{R_3(b')}{\delta} \, 2^{R_1(b')+R_2(b')-1}\\
& \qquad  + 
\hspace{-1cm} \sum_{\begin{array}{c}
a'+  R_1(b') + \delta =a\\
R(b')=b\\
0\leq a'\leq n', \, 
0\leq b'\leq \alpha\\
0\leq \delta<  R_3(b')\\
u\in A_3(b')
\end{array}} \hspace{-1cm} S'(a',b')  \,   \binom{R_3(b')-1}{\delta} \, 2^{R_1(b')+R_2(b')}
\end{align*}
and
\begin{align*} S_1(a,b) &= \hspace{-1cm}\sum_{\begin{array}{c}
a'+   R_1(b') + \delta=a\\
R(b')=b\\
0\leq a'\leq n',\, 
0\leq b'\leq \alpha\\
0\leq \delta\leq R_3(b')\\
u\notin A_3(b')
\end{array}}  \hspace{-1cm} S'(a',b')  \,    \binom{R_3(b')}{\delta} \, 2^{R_1(b')+R_2(b')-1} \\
& \qquad + 
\hspace{-1cm}\sum_{\begin{array}{c}
a'+  R_1(b') + \delta +1 =a\\
R(b')=b\\
0\leq a'\leq n', \, 
0\leq b'\leq \alpha\\
0\leq \delta<  R_3(b')\\
u\in A_3(b')
\end{array}}  \hspace{-1cm}S'(a',b')  \,   \binom{R_3(b')-1}{\delta} \, 2^{R_1(b')+R_2(b')}
\end{align*}

All the required values can be computed in polynomial time using a double scanning similar to the one used in  the proof of Lemma~\ref{lem:complete_t2}.

Thus, the claim holds.
\qqed\end{proof}

\section{Star influence models}
\label{sec:Star-Inf}
Now we  consider a family of influence graphs with a star-topology which was previously studied in the context of influence games~\cite{MRS13}.
In such a graph the two layered bipartite topology is restricted to be a star graph and extended by allowing some bidirectional connections to the center of the star. 
In a star influence graph, in addition to the sets $\tL$, $\tI$ and $\tF$, we have the central node  $c$ wich acts as mediator  and the set $\tR$ of {\em reciprocal} actors.

\begin{definition}
A {\em star influence graph} is an influence graph $(G,f)$, where $V(G)=\tL\cup\tI\cup\tR\cup\{c\}\cup\tF$ and
$E(G)=\{(u,c)\mid u\in\tL\cup\tR\} \cup \{(c,v)\mid v\in\tR\cup\tF\}$.
A {\em star influence game} is a game $\Gamma= (G,f,q,N)$, where $N=\tL\cup\tR\cup\tI$ and $(G,f)$ is a star influence graph.
\end{definition}
Without loss of generality we can assume that  the labeling function of a star influence game satisfies $f(i)\in\{0,1\}$, for $i\in V(G)\setminus\{c\}$.
Note that a reciprocal actor with a label greater than~$1$ could never be influenced by the mediator, so we remove the arcs from the center to the actor which becomes an additional opinion leader in the new graph. 
We can also assume that $f(c)\in\{0,1,\ldots,|\tL|+|\tR|+1\}$.

\begin{lemma}
\label{lem:stars}
Let $(G,f,q,N)$ be a star influence game. For $1\leq k\leq n$, the values
$|F_k(N,G,f)|$ can be computed in polynomial time.
\end{lemma}

\begin{proof}

We consider three cases and provide either a  closed formula  or a recursion allowing to compute  $|F_k(N,G,f)|$ in polynomial time.  Let $\tRF=\tR\cup \tF$.


\smallskip
\noindent 
\textbf{Case 1: }  $f(i)>0$, for $i\in V(G)$ .

When  $k<f(c)$ and $X\in F_k(N)$, we know that $c\notin F(X\cap N)$.  Thus, $F_k(N)$ only contains those sets $X$ with  $|X\cap N|=k$.
When $k\geq f(c)$,  $F_k(N)$ can be divided into two subsets: those with $c\notin F(X\cap N)$ and those with $c\in F(X\cap N)$.
If $c\notin F(X\cap N)$,  we know that $F(X)=X$ and that $|X\cap(\tL\cup\tR)|<f(c)$. To fulfill this condition it is enough  to take $0\leq i\leq f(c)-1$ vertices from $\tL\cup\tR$ and the remaining $k-i$ vertices from $\tI$ .
If $c\in F(X\cap N)$, we know that $|X\cap(\tL\cup\tR)|\geq f(c)$ and, therefore, $\tR \cup \tF \cup \{c\} \subseteq F(X\cap N)$. To attain those restrictions, together with  $F(X\cap N)=k$, we have to take $i$ vertices from $\tL$, $j$ vertices from $\tR$, for some values of $i,j$ verifying  $i+j\geq f(c)$  and $k-i-(|\tRF| +1)\geq 0$, and complete the set with $k-i-(|\tRF| +1)$ elements from $\tI$. Putting all together, we have the following expressions.\\
When  $k<f(c)$, 
\[
\frac{|F_k(N)|}{2^{|\tF|+1}}= \binom{|\tL|+|\tR|+|\tI|}{k}
\]
When  $k\geq f(c)$,
\[
\frac{|F_k(N)|}{2^{|\tF|+1}} =\sum_{i=0}^{f(c)-1}\binom{|\tL|+|\tR|}{i}\binom{|\tI|}{k-i} +
\hskip -4mm
 \sum_{\substack{i+j\geq f(c)\\
 0\leq i\leq |\tL| \ , 0\leq  j\leq |\tR|\\
  k-i- (|\tRF| +1)\geq 0}}
\hskip -4mm\binom{|\tL|}{i}\binom{|\tR|}{j}\binom{|\tI|}{k-i-(|\tRF|+1)}
\]

\smallskip
\noindent 
\textbf{Case 2: }   $f(c)=0$.

Observe that in this case, for $X\subseteq V$, $F(X\cap N)=(X\cap (\tL\cup \tI)) \cup \tR \cup \tF\cup \{c\}$. Thus either  $k<|\tRF|+1$ or $k\geq |\tRF|+1$ .
In the first case, there are no sets expanding to $k$ vertices. In the second,   we have  to select the adequate number of vertices from $\tL\cup \tI$ and any number from $\tR$. This leads to the following expression:
\[
\frac{|F_k(N)|}{2^{|\tF|+1}} = 
\begin{dcases}
 0                                                 & \mbox{ if } k<|\tRF|+1\\
 \binom{|\tL|+|\tI|}{k-(|\tRF|+1)}\ 2^{|\tR|} & \mbox{ if } k\geq |\tRF|+1.
\end{dcases}
\]

\smallskip
\noindent 
\textbf{Case 3: }  $f(c)>0$ and, for some $i\in V(G)\setminus\{c\}$, $f(i)=0$.

Let $Z_1=\{i\in\tI\cup\tF\mid f(i)=0\}$, $Z_2=\{i\in\tL\cup\tR\mid f(i)=0\}$, $z_1=|Z_1|$ and $z_2=|Z_2|$.

Let $X\subseteq V$. Observe that, for  $i\in Z_1\cup Z_2$, $i\in F(X\cap N)$.  Therefore, we can remove those vertices from the influence graph taking care of reducing the label of  $c$ whenever a vertex from $Z_2$ is removed and reducing the size of the required expansion.  So, we construct the influence graph $(G',f')$ where
$G'=G\setminus (Z_1\cup Z_2)$ and $f'(u)=f(u)$, for $u\neq c$, and $f'(c)=\max\{f(c)-z_2, 0\}$. We have 
$$|F_k(N,G,f)|=|F_{k-z_1-z_2}(N\setminus (Z_1 \cup Z_2),G',f')| 2^{z_1+z_2}.$$ 
As, only $f'(c)$ can be zero, the later expression can be computed using the formulas provided in the previous cases and the claim follows.  
\qqed\end{proof}

It is known that the problem of counting the number of winning or losing coalitions of a given influence game is \#P-complete~\cite{MRS15}. The expressions provided in Lemma~\ref{lem:stars} allow us  to count the number of winning and losing coalitions, for  star influence games, in polynomial time, by computing $\sum_{i=q}^n|F_i(N)|$ and $\sum_{i=0}^{q-1}|F_i(N)|$, respectively.

\begin{example}
Consider the star influence graph $(G,f)$ of Figure~\ref{fig:star} and the star influence model $(G,f,4,N)$. Here $|\tI|=|\tF|=1$, $|\tR|=2$, $|\tL|=3$ and $f(c)=3$.
Hence, $|\cW|=\sum_{i=4}^8|F_i(N)|=4 (0+3+12+13+4)=128$, and $|\cL|=\sum_{i=0}^3|F_i(N)|=4 (1+6+15+10)=128$. Note that it holds that $|\cW|+|\cL|=256=2^8$, as expected.
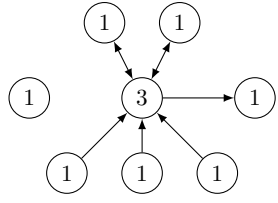
\begin{figure}[t]
\centering
\begin{tikzpicture}[every node/.style={circle,scale=0.8}, >=latex]
\node[draw](a) at (1.0,2.0)[] {1};
\node[draw](b) at (2.0,2.0)[] {1};
\node[draw](c) at (0.0,1.0)[] {1};
\node[draw](d) at (1.5,1.0)[] {3};
\node[draw](e) at (3.0,1.0)[] {1};
\node[draw](f) at (0.5,0.0)[] {1};
\node[draw](g) at (1.5,0.0)[] {1};
\node[draw](h) at (2.5,0.0)[] {1};
\draw[<->](a) to node {}(d);
\draw[<->](b) to node {}(d);
\draw[->] (d) to node {}(e);
\draw[->] (f) to node {}(d);
\draw[->] (g) to node {}(d);
\draw[->] (h) to node {}(d);
\end{tikzpicture}
\caption{A star influence graph.\label{fig:star}}
\end{figure}
\end{example}

In order to solve the \SAT problem for star influence models we first show that the \EXPN problem can be solved in polynomial time in a slightly extended family of influence graphs.
 
\begin{definition}
An {\em extended star influence graph} is an influence graph that is obtained from a star influence graph $(G',f')$, by selecting one vertex $u\in \tR(G')$ and adding a set of vertices $F_u$ with label 1 and the set of edges $\{(u,v)\mid v\in F_u\}$.
An {\em extended star influence model} is an influence model $(G,f,q,N)$, where $(G,f)$ is an extended star influence graph and $N=\tL\cup\tR\cup\tI$.
\end{definition}

Observe that all the additional vertices attached to a reciprocal actor are followers. By taking $F_u=\emptyset$ we obtain a star influence graph.

\begin{lemma}
Let $(G,f)$ be an extended star influence graph and let $N=\tL(G)\cup\tI(G)\cup\tR(G)$.
For $1\le{}k\le{}N$, $|F_k(N,G,f)|$ can be computed in polynomial time.
\end{lemma}

\begin{proof}
Let $(G,f)$ be an extended star influence graph.
Let $u$ be the selected vertex in $\tR$. Assume that, for $i\in V(G)$,  $f(i)>0$.
Set $\tLR= \tL\cup \tR$ and $\tRF=\tR\cup \tF$.

If $k<f(c)$ and $X\in F_k(N)$,  we know that $c\notin F(X\cap N)$. Therefore, 
$F_k(N)$ only contains those sets $X$ with  $|X\cap N|=k$.
If $u\notin X$, then $F(X\cap N)=X\cap N$; but if $u\in X$, then $F(X\cap N)=(X\cap N)\cup\tF_u$.

On the other hand, if $k\geq f(c)$,  the set $F_k(N)$ can be divided into two subsets, those with $c\notin F(X\cap N)$ and those with $c\in F(X\cap N)$.
As in the proof of Lemma~\ref{lem:stars}, when $c\notin F(X\cap N)$,  we need $|X\cap (\tL\cup \tR)|<f(c)$;
but when $c\in F(X\cap N)$, $\tR\cup\tF\cup\tF_u\subseteq F(X\cap N)$. Therefore, we have the following expressions.\\
When  $k<f(c)$,
\[
\frac{|F_k(N)|}{2^{|\tF|+|\tF_u|+1}} =
\binom{|\tLR|-1+|\tI|}{k} + \binom{|\tLR|-1+|\tI|}{k-|\tF_u|-1}.
\]
When  $k\geq f(c)$,
\begin{align*}
\frac{|F_k(N)|}{2^{|\tF|+|\tF_u|+1}} & = \sum_{i=0}^{f(c)-1}\left[\binom{|\tLR|-1}{i}\binom{|\tI|}{k-i} 
+  \binom{|\tLR|-1}{i}\binom{|\tI|}{k-i -|\tF_u|}\right]\\
& \qquad + \hskip -4mm
 \sum_{\substack{i+j\geq f(c)\\
 0\leq i\leq |\tL| \ , 0\leq  j\leq |\tR|\\
  k-i- (|\tRF| +1)\geq 0}}\binom{|\tL|}{i}\binom{|\tR|}{j}\binom{|\tI|}{k-i-(|\tR|+|\tF|+|\tF_u|+1)}. 
\end{align*}

Using this expression, $|F_k(N)|$ can be computed in polynomial time.
An argument  similar to the one in the proof of Lemma~\ref{lem:stars} allows us to 
devise a polynomial time algorithm to compute $|F_k(N)|$ when some of the labels are zero.
\qqed\end{proof}

We  transfer the previous results to an algorithm for solving the \SAT problem for the corresponding oblivious models. Note that, for a given star influence model $(G,f,q,N)$, the graphs $R(G,f,i)$ and $R_2(G,f,i)$, as required in Lemmas~\ref{lem:SAT_for_oblivious2} and~\ref{lem:SAT_for_oblivious3}, are extended star influence graphs. As a consequence of those results  we have the following.

\begin{theorem}\label{teo:sat-star-o}
The \SAT problem, for  oblivious models  corresponding to star influence games, is polynomial time solvable.
\end{theorem}

Finally, we show  how to solve the \SAT problem in non-oblivious models.

\begin{theorem}
The \SAT problem, for  non-oblivious  models  corresponding to star influence games, is polynomial time solvable.
\end{theorem}
\begin{proof}
Let $\Gamma= (G,f,q,N)$ be a star influence game.   We analyze the differences among participants in the oblivious $\cM^o$ and non-oblivious $\cM^n$ associated models.
Recall that,   for $x\in\{0,1\}^n$, $p_i(x)=|F(X(x)\cap N)\cap P(i)|$ and $q_i(x)=|P(i)\setminus F(X(x)\cap N)|$ and that the final decision vectors in both models are defined as follows.

In the oblivious model, for $i\in V(G)$,
\[c_i ^{\cM^o}= \begin{cases}
 1 & \mbox{if } i\in F(X(x)),\\
 0 & \mbox{otherwise.}\\
\end{cases}
\]

In the non-oblivious model we have two cases, for $i\in V(G)\setminus N$, 
\[c_i^{\cM^n} = \begin{cases}
 1  & \mbox{if } p_i(x) \geq f(i) \text{ and } q_i(x) < f(i),\\
 0  & \mbox{if } q_i(x) \geq f(i) \text{ and } p_i(x) < f(i),\\
x_i & \mbox{otherwise,}\\
\end{cases}
\]
and, for $ i\in N$,
$$c_i ^{\cM^n}= \begin{cases}
 1 & \mbox{if } i\in F(X(x)),\\
 0 & \mbox{otherwise.}\\
\end{cases}
$$

Let $x\in\{0,1\}^n$ be an initial decision vector.  Note that,  for $i\in V\setminus \{c\}$, $f(i)\in\{0,1\}$ and the in-degree of  $i$ is either 0 or 1. Thus, according to the above expressions,  for $i\in V\setminus \{c\}$, $c^{\cM^o}_i(x)=c^{\cM^n}_i(x)$.  This implies that the final decision vectors in $\cM^o$ and $\cM^n$ can only differ in the final decision of $c$.  

When $p_c(x),q_c(x)\geq f(i)$, $c^{\cM^o}_c(x)=1$  but $c^{\cM^n}_c(x)=x_c$.
When $p_c(x),q_c(x)<f(i)$, $c^{\cM^o}_c(x)=0$  but $c^{\cM^n}_c(x)=x_c$. In all the remaining cases we have   $c^{\cM^o}_c(x)=c^{\cM^n}_c(x)$.

The different final decision of $c$ has relevance only if it implies a  change in the final collective decision, therefore we have to examine only two cases in which a difference can arise.

\noindent
\textbf{Case 1} $p_c(x),q_c(x)\geq f(i)$  and  $|F(X(x))|=q$. 

In this case $c_{\cM^o}(x) =1$, thus in the oblivious model $c$ is satisfied only when $x_c=1$.
But, $c_{\cM^n}(x) =x_c$, thus in the non-oblivious model $c$ is satisfied independently of its initial choice.

\noindent
\textbf{Case 2} $p_c(x),q_c(x)< f(i)$  and  $|F(X(x))|=q-1$.

In this case $c_{\cM^o}(x) =0$, thus in the oblivious model $c$ is satisfied only when $x_c=0$.
But, $c_{\cM^n}(x) =x_c$, thus in the non-oblivious model $c$ is satisfied independently of its initial choice.
 
Using the above properties, we can obtain an expression for   $\SatM_{\cM^n}(i)$.  

When $i\neq C$, the final collective decision in both $\cM^o$ and $\cM^n$ is independent of the value of $x_i$. Therefore,  $\SatM_{\cM^n}(i)=\SatM_{\cM^o}(i)= 2^{n-1}$.

For the central vertex $c$ we have, 
\begin{align*}
\SatM_{\cM^n}(c)=&\SatM_{\cM^o}(c) \\
& + |\{x\in\{0,1\}^n\mid x_c=0, \, p_c(x),q_c(x)\geq f(c) \text{ and } |F(X(x))|=q\}|\\
& + |\{x\in\{0,1\}^n\mid x_c=1, \, p_c(x),q_c(x)< f(c) \text{ and } |F(X(x))|=q-1\}|.
\end{align*}

It is easy to derive closed formulas for the sizes of the sets appearing in the above expressions. Therefore, from Theorem~\ref{teo:sat-star-o} the claim follows. 
\qqed\end{proof}
\section{Conclusions and open problems}

We have introduced collective decision-making models that extend the opinion leader-follower decision models to the setting of influence games. Our \gOLF models extend the original model proposed by van den Brink et al.~\cite{BRS11} to include decision rules different from the simple majority rule. The rules of influence spread are different for \OLF and  for influence games. This leads us to the definition of oblivious and non-oblivious influence models. We have established a connection between players in an influence game and opinion leaders in the associated decision model.

We have studied the computational complexity of the satisfaction measure proposed by van den Brink et al.~\cite{BRS11} for the new collective decision-making models. We have shown that computing the satisfaction measure is \#P-hard even for the more restrictive family of \oOLF models. The problem remains hard when the influence graph is restricted to be a two layered bipartite graph with constant in-degree. Interestingly enough, the measure coincides with the well established \Rae index that is closely related to the Banzhaf value of the associated simple games. Therefore we have found another set of simple games with succinct representation in which computing the Banzhaf value is \#P-hard.

To complement the hardness result we have explored the possibility of having models in which the \SAT problem is tractable by making stronger connections among sets of vertices.
We defined two subfamilies of oblivious influence models, the strong hierarchical influence models and the star influence models in which satisfaction can be computed in polynomial time. The first family considers multilayered bipartite influence graphs having a hierarchy of mediators moderating the spread of influence from opinion leaders to followers. In the  star influence models we also allow a restricted form of mutual influence. The result for  star influence models can be extended, by a tedious proof, following the table computation for hierarchical influence graphs, to a more general version of star influence graphs in which the central vertex is replaced by a set of independent vertices. Any of those vertices is a replica of the central vertex, thus they keep all the existing connections of the central vertex to the remaining vertices.

For both strong hierarchical influence models and the  star influence models, besides computing in polynomial time the satisfaction measure we can compute additional information. In particular we can compute in polynomial time the number of initial decision vectors that expand to $k$ actors being $b$ of them followers. We can consider a {\em mixed oblivious influence model} whose structure is obtained from an influence graph $(G,f)$ and a strong hierarchical or a star influence graph $(G',f')$ adding a complete bipartite graph between $V(G')$ and a set $U\subseteq V(G)$. It is easy to show that the $F_k(\tL(G')\cup \tI(G') \cup \tR(G'))$ in a mixed oblivious influence model can be computed in polynomial time. As the connection is by means of a complete graph only a polynomial number of subsets of $U$ can be activated. Thus strong hierarchical or  star influence graph can be used as mechanism to control the spread of influence in a generic social network.

Among several problems that remain open we want to point out two that relate directly to the hardness proof provided in this paper.  The first one is the complexity of the \SAT problem for \OLF models under the simple majority rule as our reduction does not construct an \oOLF with this restriction. Although we conjecture that the problem is hard we have been unable to straighten the reduction. The second one relates to other aspects of the complexity of counting independent sets of size $\frac{1}{3}n$. In~\cite{Hof10} was shown a stronger result, 
 as the provided reduction from \#3{\sc-SAT}, give a formula with $m$ clauses, constructs a graph
with ${\cal O}(m)$ vertices  in such a way that the number of satisfying assignments
of the original formula coincides with the number of independent sets of size $\frac{1}{3}|V|$ of the constructed graph. Such property implies that the problem of computing the number of independent sets with size exactly $\frac{2}{3}|V|$ in a graph is harder, in the sense that it cannot be solved by a sub-exponential time algorithm, unless the $\#3${\sc-Sat} problem can be solved in sub-exponential time~\cite{Hof10}.  The graph constructed in our reduction does not have linear size with respect to the input graph.   Hence, we can only deduce that the problem is \#P-hard. It remains open to show if sub-exponential time algorithms can be also ruled out. 

\remove{One restrictive decision model that has received some attention is the one in which the unanimity of the predecessors is required to adopt a decision. This is equivalent to set $r=1$ in the \OLF models or set $f(u)=\delta^-(u)$ in influence models.  Van den Brink~et al.   provided an axiomatization for  the satisfaction measure  for \OLF{s} with $r=1$ \cite{BRS12}.  The influence games on undirected influence graphs with  $f(u)=\delta^-(u)$ received the name of \emph{maximum influence} games \cite{MRS15}. Interestingly enough several problems that are hard for influence games became tractable for maximum influence games. None of the reductions  in this paper  produces a model in which unanimity of predecessors  is required. Therefore, the complexity of computing the satisfaction measure in  such subfamilies of decision models remains as an interesting open problem.  }

Another line for further work is to study the computational complexity of other measures for collective decision-making models. 
In particular  the {\em power} measure introduced in~\cite{BRS11}. The power of an actor  is the number of society's initial decisions for which that actor can change the collective decision by changing its initial  decision. 
As the satisfaction measure, the power measure  can be expressed in a form similar to the  Rae index, but considering the sets of minimal winning and maximal losing coalitions, instead of $\cW$ and $\cL$. 
The satisfaction measure is closely related to the {\em Chow parameters} for simple games~\cite{DS79}, but considering also losing coalitions. Chow parameters were initially defined in the 1960s in the context of threshold Boolean functions~\cite{Cho61}, and it is known that their computation, for simple games given by their minimal winning coalitions, is \#P-complete~\cite{Azi08}. In a similar way,  the power measure is related with another power index called the {\em Holler index}~\cite{Azi09}. It will be of interest to see if the results on this paper can be extended to the computation of those and other measures.

Finally, let us mention that due to the relationship that we have established between these models and simple games, all the power indices introduced in the context of simple games can also be transformed in measures for influence models. It is of interest to analyze which of those power indices provide interesting measures for decision systems.


\section*{Acknowledgments}

Xavier Molinero has been partially supported by funds from the Spanish Ministry of Economy and Competitiveness (MINECO) under grant MTM2015-66818-P.  Fabi\'an Riquelme was partially supported by project PMI USA1204. Maria Serna has been partially supported by funds from the Spanish Ministry of Economy and Competitiveness (MINECO) and the European Union (FEDER funds) under grant COMMAS (ref. TIN2013-46181-C2-1-R), and SGR 2014--1034  (ALBCOM) of the Catalan government.


\begin{thebibliography}{10}

\bibitem{Azi08}
H.~Aziz.
\newblock Complexity of comparison of influence of players in simple games.
\newblock In Ulle Endriss and Paul~W. Goldberg, editors, {\em Proceedings of
  the 2nd International Workshop on Computational Social Choice, (COMSOC
  2008)}, pages 61--72, 2008.

\bibitem{Azi09}
H.~Aziz.
\newblock {\em Algorithmic and complexity aspects of simple coalitional games}.
\newblock PhD thesis, Department of Computer Science, University of Warwick,
  2009.

\bibitem{Ban65}
J.~F. Banzhaf.
\newblock Weighted voting doesn't work.
\newblock {\em Rutgers Law Review}, 19:317--343, 1965.

\bibitem{Cho61}
C.~K. Chow.
\newblock On the characterization of threshold functions.
\newblock In {\em 2nd Annual Symposium on Switching Circuit Theory and Logical
  Design, Detroit, Michigan, USA, October 17-20, 1961}, pages 34--38. IEEE
  Computer Society, 1961.

\bibitem{DP94}
X.~Deng and C.~Papadimitriou.
\newblock On the complexity of cooperative solution concepts.
\newblock {\em Mathematics of Operations Research}, 19(2):257--266, 1994.

\bibitem{DS79}
P.~Dubey and L.~S. Shapley.
\newblock Mathematical properties of the banzhaf power index.
\newblock {\em Mathematics of Operations Research}, 4(2):99--131, 1979.

\bibitem{GJ79}
M.~R. Garey and D.~S. Johnson.
\newblock {\em Computers and intractability: A guide to the theory of
  {NP}--Completness}.
\newblock A Series of Books in the Mathematical Sciences. W.~H. Freeman and
  Company, New York, NY, 1979.

\bibitem{Gra78}
M.~Granovetter.
\newblock Threshold models of collective behavior.
\newblock {\em American Journal of Sociology}, 83(6):1420--1443, 1978.

\bibitem{Hof10}
C.~Hoffmann.
\newblock Exponential time complexity of weighted counting of independent sets.
\newblock In Venkatesh Raman and Saket Saurabh, editors, {\em Parameterized and
  Exact Computation - 5th International Symposium, {IPEC} 2010, Chennai, India,
  December 13-15, 2010. Proceedings}, volume 6478 of {\em Lecture Notes in
  Computer Science}, pages 180--191. Springer, 2010.

\bibitem{KL55}
E.~Katz and P.~F. Lazarsfeld.
\newblock {\em Personal influence: The part played by people in the flow of
  mass communication}.
\newblock Foundations of communications research. Free Press, Glencoe, IL,
  1955.

\bibitem{KKT03}
D.~Kempe, J.~Kleinberg, and \'E. Tardos.
\newblock Maximizing the spread of influence through a social network.
\newblock In L.~Getoor, T.~E. Senator, P.~Domingos, and C.~Faloutsos, editors,
  {\em Proceedings of the Ninth ACM SIGKDD International Conference on
  Knowledge Discovery and Data Mining, Washington, DC, USA, August 24 - 27,
  2003}, pages 137--146, 2003.

\bibitem{LMF06}
A.~Laruelle, R.~MartÄ±nez, and F.~Valenciano.
\newblock Success versus decisiveness: Conceptual discussion and case study.
\newblock {\em Journal of Theoretical Politics}, 18(2):185--205, 2006.

\bibitem{LBG68}
P.~F. Lazarsfeld, B.~Berelson, and H.~Gaudet.
\newblock {\em The people's choice - {H}ow the voter makes up his mind in a
  presidential campaign}.
\newblock Columbia University Press, New York, NY, 3rd edition, 1968.
\newblock First edition published in 1944.

\bibitem{MRS13}
X.~Molinero, F.~Riquelme, and M.~J. Serna.
\newblock Star-shaped mediation in influence games.
\newblock In K.~Cornelissen, R.~Hoeksma, J.~Hurink, and B.~Manthey, editors,
  {\em 12th Cologne-Twente Workshop on Graphs and Combinatorial Optimization,
  Enschede, Netherlands, May 21-23, 2013}, volume WP~13-01 of {\em CTIT
  Workshop Proceedings}, pages 179--182, 2013.

\bibitem{MRS15}
X.~Molinero, F.~Riquelme, and M.~J. Serna.
\newblock Cooperation through social influence.
\newblock {\em European Journal of Operation Research}, 242(3):960--974, 2015.

\bibitem{MRS15b}
X.~Molinero, F.~Riquelme, and M.~J. Serna.
\newblock Forms of representations for simple games: sizes, conversions and
  equivalences.
\newblock {\em Mathematical Social Sciences}, 76:87--102, 2015.

\bibitem{Pen46}
L.~Penrose.
\newblock The elementary statistics of majority voting.
\newblock {\em Journal of the Royal Statistical Society}, 109(1):53--57, 1946.

\bibitem{Rae69}
D.~W. Rae.
\newblock Decision-rules and individual values in constitutional choice.
\newblock {\em The American Political Science Review}, 63(1):40--56, 1969.

\bibitem{Sch78}
T.~Schelling.
\newblock {\em Micromotives and macrobehavior}.
\newblock Fels lectures on public policy analysis. W.~W. Norton \& Company, New
  York, NY, 1978.

\bibitem{TZ99}
A.~Taylor and W.~Zwicker.
\newblock {\em Simple games: Desirability relations, trading,
  pseudoweightings}.
\newblock Princeton University Press, Princeton, NJ, 1999.

\bibitem{Tro66}
V.~C. Troldahl.
\newblock A field test of a modified ``two-step flow of communication'' model.
\newblock {\em Public Opinion Quarterly}, 30(4):609--623, 1966.

\bibitem{BRS11}
R.~van~den Brink, A.~Rusinowska, and F.~Steffen.
\newblock Measuring power and satisfaction in societies with opinion leaders:
  dictator and opinion leader properties.
\newblock {\em Homo Oeconomicus}, 28(1--2):161--185, 2011.


\bibitem{vNM44}
J.~von Neumann and O.~Morgenstern.
\newblock {\em Theory of games and economic behavior}.
\newblock Princeton University Press, Princeton, NJ, 1944.

\end{thebibliography}

\end{document}